%% file: main.tex
\documentclass[a4paper,singlecolor,11pt]{article}

\usepackage[margin=1in]{geometry}

\usepackage{lmodern}
\usepackage{amssymb}
\usepackage{amsmath}
\usepackage{graphicx}
\usepackage{subcaption}
\usepackage{makeidx}
\usepackage{multicol}

\usepackage{changepage}

\makeindex

\usepackage{version}

\newcommand{\framework}{\textsc{ParButterfly}\xspace}

\usepackage{xspace}
\usepackage{algorithmicx}
\usepackage{graphicx,url}
\usepackage{multirow,multicol,mathtools}
\usepackage[noend]{algpseudocode}
\usepackage{balance}
\usepackage{paralist}
\usepackage{booktabs}
\usepackage{color}
\usepackage{rotating}
\usepackage{enumitem}
\usepackage[most]{tcolorbox}
\usepackage{amsthm,amssymb,stmaryrd}

\newtheorem{theorem}{Theorem}[section]

\newtheorem{lemma}[theorem]{Lemma}

\newcommand{\myparagraph}[1]{\medskip\noindent {\bf #1.}}

\newcommand{\whp}[1]{\emph{whp}}
\newcommand{\defn}[1]{{\bf{\emph{#1}}}}

\newcommand{\rank}[1]{{\small\textsf{rank}(#1)}}

\usepackage[ruled]{algorithm}

\newcommand{\BigO}[1]{\ensuremath{\operatorname{O}\bigl(#1\bigr)}}
\newcommand{\BigOmega}[1]{\ensuremath{\operatorname{\Omega}\bigl(#1\bigr)}}

\newcommand{\slfrac}[2]{\left.#1\middle/#2\right.}

\newcommand{\algname}[1]{\textnormal{\textsc{#1}}}

\algblock{ParFor}{EndParFor}
\algnewcommand\algorithmicparfor{\textbf{parfor}}
\algnewcommand\algorithmicpardo{\textbf{do}}
\algnewcommand\algorithmicendparfor{}
\algrenewtext{ParFor}[1]{\algorithmicparfor\ #1\ \algorithmicpardo}
\algrenewtext{EndParFor}{\algorithmicendparfor}
\algtext*{EndParFor}

\algdef{SE}[DOWHILE]{Do}{doWhile}{\algorithmicdo}[1]{\algorithmicwhile\ #1}

\begin{document}

\title{\Large Parallel Algorithms for Butterfly Computations}
\author{Jessica Shi \\ MIT CSAIL \\ jeshi@mit.edu
\and Julian Shun \\ MIT CSAIL \\ jshun@mit.edu}

\date{}

\maketitle

\input{abstract}

\input{intro}
\input{prelims}

\input{framework}
\input{analysis}

\input{app}

\input{eval}

\input{related}
\input{conclusion}

\section*{Acknowledgements}
We thank Laxman Dhulipala for helpful discussions about bucketing.
This research was supported by NSF Graduate Research Fellowship
\#1122374, DOE Early Career Award \#DE-SC0018947, NSF CAREER Award
\#CCF-1845763, DARPA SDH Award \#HR0011-18-3-0007, and Applications
Driving Architectures (ADA) Research Center, a JUMP Center
co-sponsored by SRC and DARPA.

\bibliographystyle{plain}
\bibliography{references}

\end{document}

%% file: abstract.tex
\begin{abstract} 
  
Butterflies are the smallest non-trivial subgraph in bipartite graphs,
and therefore having efficient computations for analyzing them is
crucial to improving the quality of certain applications on bipartite
graphs. In this paper, we design a framework called \framework that
contains new parallel algorithms for the following problems on
processing butterflies: global counting, per-vertex counting, per-edge
counting, tip decomposition (vertex peeling), and wing decomposition
(edge peeling). The main component of these algorithms is aggregating
wedges incident on subsets of vertices, and our framework supports
different methods for wedge aggregation, including sorting, hashing,
histogramming, and batching. 
In addition, \framework supports different ways of ranking the
vertices to speed up counting, including side ordering, approximate
and exact degree ordering, and approximate and exact complement
coreness ordering.  For counting, \framework also supports both exact
computation as well as approximate computation via graph
sparsification.  We prove strong theoretical guarantees on the work
and span of the algorithms in \framework.

We perform a comprehensive evaluation of all of the algorithms
in \framework on a collection of real-world bipartite
graphs using a 48-core machine. Our counting algorithms obtain
significant parallel speedup, outperforming the fastest sequential
algorithms by up to 13.6x with a self-relative speedup of up to 38.5x.
Compared to general subgraph counting solutions, we are orders of
magnitude faster. Our peeling algorithms achieve self-relative
speedups of up to 10.7x and outperform the fastest sequential baseline
by up to several orders of magnitude.

This is an extended version of the paper of the same name that appeared
in the \emph{SIAM Symposium on Algorithmic Principles of Computer
  Systems, 2020}~\cite{Shi2020}.

\end{abstract}

%% file: intro.tex
\section{Introduction}\label{sec:intro}

A fundamental problem in large-scale network analysis is finding and
enumerating basic graph motifs. Graph motifs that represent the
building blocks of certain networks can reveal the underlying
structures of these networks. Importantly, triangles are core
substructures in unipartite graphs, and indeed, triangle counting is a
key metric that is widely applicable in areas including social network
analysis~\cite{Newman03}, spam and fraud detection~\cite{BeBoCaGi08},
and link classification and recommendation~\cite{TsDrMiKoFa11}. 
 However, many real-world graphs are bipartite and model the
 affiliations between two groups.  For example,
 bipartite graphs are used to represent peer-to-peer exchange networks
 (linking peers to the data they request), group membership
 networks (e.g., linking actors to movies they acted in),
 recommendation systems (linking users to items they rated),
 factor graphs for error-correcting codes, and
 hypergraphs~\cite{BoEv97, LaMaVe08}.  Bipartite graphs
 contain no triangles; the smallest non-trivial subgraph is a
 \defn{butterfly} (also known as rectangles), which is a $(2,
 2)$-biclique (containing two vertices on each side and all four
 possible edges among them), and thus having efficient algorithms
 for counting butterflies is crucial for applications on bipartite
 graphs~\cite{WaFuCh14,AkKoPi17,SaSaTi18}. Notably, butterfly counting 
 has applications in link spam detection \cite{GiKuTo05} 
 and document clustering \cite{Dhillon01}. Moreover, butterfly
 counting naturally lends itself to finding dense subgraph structures
 in bipartite networks. Zou \cite{Zou16} and Sariy\"{u}ce and
 Pinar~\cite{SaPi18} developed peeling algorithms to hierarchically
 discover dense subgraphs, similar to the $k$-core decomposition for
 unipartite graphs~\cite{Seidman83, MaBe83}.  An example bipartite
 graph and its butterflies is shown in
 Figure~\ref{img_butterfly}.

There has been recent work on designing efficient sequential
algorithms for butterfly counting and peeling~\cite{ChNi85,
  WaFuCh14,Zhu2018,Zou16, SaSaTi18, SaPi18, wang2020efficient}. However, given the high
computational requirements of butterfly computations, it is natural to
study whether we can obtain performance improvements using parallel
machines. This paper presents a framework for butterfly computations,
called \framework, that enables us to obtain new parallel algorithms
for butterfly counting and peeling. \framework is a modular framework
that enables us to easily experiment with many variations of our
algorithms.  We not only show that our algorithms are efficient in
practice, but also prove strong theoretical bounds on their work and
span. 
Given that all real-world bipartite graphs fit on a multicore machine,
we design parallel algorithms for this setting.

\begin{figure}[!t]
\centering
\includegraphics[width=0.3\textwidth, page=3]{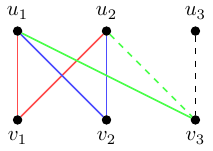}
\caption[caption]{The butterflies in this graph are $\{ u_1, v_1, u_2, v_2\}$, $\{ u_1, v_1, u_2, v_3\}$,  and $\{ u_1, v_2, u_2, v_3\}$. The red, blue, and green edges each produce a wedge ($\{u_1, v_1, u_2\}$, $\{ u_1,v_2,u_2\}$, and $ \{ u_1,v_3,u_2\}$). Note that these wedges all share the same endpoints, namely $u_1$ and $u_2$. Thus, any combination of two of these wedges forms a butterfly. For example, $\{u_1, v_1, u_2\}$ and $\{ u_1,v_2,u_2\}$ combine to form a butterfly $\{ u_1, v_1, u_2, v_2\}$. However, the dashed edges produce another wedge, $\{ u_2, v_3, u_3 \}$, which has different endpoints, namely $u_2$ and $u_3$. This wedge cannot be combined with any of the previous wedges to form a butterfly. }
\label{img_butterfly}
\end{figure}

For butterfly counting, the main procedure involves finding wedges
($2$-paths) and combining them to count butterflies. See
Figure~\ref{img_butterfly} for an example of wedges. In particular, we
want to find all wedges originating from each vertex, and then
aggregate the counts of wedges incident to every distinct pair of
vertices forming the endpoints of the wedge. With these counts, we
can obtain global, per-vertex, and per-edge 
butterfly counts. The \framework framework provides
different ways to aggregate wedges in parallel, including 
sorting, hashing, histogramming, and batching.  Also, we can speed
up butterfly counting by ranking vertices and only
considering wedges formed by a particular ordering of the vertices.
\framework supports different parallel ranking methods,
including side-ordering, approximate and exact degree-ordering, and
approximate and exact complement-coreness ordering. These orderings
can be used with any of the aggregation methods.  To
further speed up computations on large graphs, \framework also
supports parallel approximate butterfly counting via graph 
sparsification based on ideas by Sanei-Mehri \emph{et al.}~\cite{SaSaTi18}
for the sequential setting. Furthermore, we integrate into \framework a recently proposed cache optimization for butterfly counting by Wang \emph{et al.}~\cite{Wang2019}.

In addition, \framework provides parallel algorithms for peeling
bipartite networks based on sequential dense subgraph discovery
algorithms developed by Zou~\cite{Zou16} and Sariy\"{u}ce and
Pinar~\cite{SaPi18}. Our peeling algorithms iteratively remove the
vertices (tip decomposition) or edges (wing decomposition) with the
lowest butterfly count until the graph is empty. Each iteration
removes vertices (edges) from the graph in parallel and updates the
butterfly counts of neighboring vertices (edges) using the parallel
wedge aggregation techniques that we developed for counting.  We use a
parallel bucketing data structure by Dhulipala \emph{et
  al.}~\cite{DhBlSh17} and a new parallel Fibonacci heap to
efficiently maintain the butterfly counts.

We prove theoretical bounds showing that some variants of our counting
and peeling algorithms are highly parallel and match the work of the
best sequential algorithm. For a graph $G(V,E)$ with $m$ edges and
arboricity $\alpha$,\footnote{The arboricity of a graph is defined to
  be the minimum number of disjoint forests that a graph can be
  partitioned into.}  \framework gives a counting algorithm that takes
$\BigO{\alpha m}$ expected work, $\BigO{\log m}$ span with high
probability (w.h.p.),\footnote{By ``with high probability'' (w.h.p.),
  we mean that the probability is at least $1 - \slfrac{1}{n^c}$ for
  any constant $c > 0$ for an input of size $n$.} and $\BigO{\min(n^2, \alpha m)}$ additional space.  Moreover, we design a parallel
  Fibonacci heap that improves upon the work bounds for vertex-peeling from Sariy\"{u}ce and Pinar's sequential
algorithms, which take work proportional to the maximum number of
per-vertex butterflies. \framework gives a vertex-peeling
algorithm that takes $\BigO{\min({\normalfont \text{max-b}}_v,\allowbreak\rho_v
  \log n) +\allowbreak\sum_{v\in V}\text{deg}(v)^2}$ expected work, 
$\BigO{\rho_v\log^2 n}$ span w.h.p., and $\BigO{n^2 + \allowbreak {\normalfont \text{max-b}}_v}$ additional space, and an edge-peeling algorithm
that takes $\BigO{\min({\normalfont \text{max-b}}_e,\allowbreak\rho_e \log n)
  +\sum_{(u,v) \in E} \sum_{u' \in N(v)} \min(\text{deg}(u),
  \allowbreak \text{deg}(u'))}$ expected work, $\BigO{\rho_e\log^2
  m}$ span w.h.p., and $\BigO{m + \allowbreak {\normalfont \text{max-b}}_e}$ additional space, where $\text{max-b}_v$ and $\text{max-b}_e$ are the
maximum number of per-vertex and per-edge butterflies and $\rho_v$ and
$\rho_e$ are the number of vertex and edge peeling iterations required
to remove the entire graph. Additionally, given a slightly relaxed work bound, we can improve the space 
bounds in both algorithms; specifically, we have a vertex-peeling algorithm that takes 
$\BigO{\rho_v \log n +\allowbreak\sum_{v\in V}\text{deg}(v)^2}$ expected work, 
$\BigO{\rho_v\log^2 n}$ span w.h.p., and $\BigO{n^2}$ additional space, and we have an edge-peeling 
algorithm that takes $\BigO{\rho_e \log n
  +\sum_{(u,v) \in E} \sum_{u' \in N(v)} \min(\text{deg}(u),
  \allowbreak \text{deg}(u'))}$ expected work, $\BigO{\rho_e\log^2
  m}$ span w.h.p., and $\BigO{m}$ additional space.

Moreover, we demonstrate further work-space tradeoffs with a vertex-peeling algorithm 
that takes $\BigO{\rho_v\log n \allowbreak + b}$ expected work, $\BigO{\rho_v\log^2 n}$ span w.h.p., and 
$\BigO{\alpha m}$ additional space, where $b$ is the total number of butterflies. Similarly, we give an edge-peeling algorithm 
that takes $\BigO{\rho_e\log n \allowbreak + b}$ expected work, $\BigO{\rho_e \log^2 n}$ span w.h.p., and 
$\BigO{\alpha m}$ additional space. We can improve the work complexities to $\BigO{b}$ expected work by allowing $\BigO{\alpha m + \allowbreak {\normalfont \text{max-b}}_v}$ and $\BigO{\alpha m + \allowbreak {\normalfont \text{max-b}}_e}$ additional space for vertex-peeling and edge-peeling respectively.

We present a comprehensive experimental evaluation of all of the different
variants of counting and peeling algorithms in the \framework
framework.  On a 48-core machine, our counting algorithms achieve
self-relative speedups of up to 38.5x and outperform the fastest
sequential baseline by up to 13.6x.  Our peeling algorithms achieve
self-relative speedups of up to 10.7x and due to their improved work
complexities, outperform the fastest sequential baseline by up to
several orders of magnitude.  Compared to PGD~\cite{AhmedNRDW17}, a
state-of-the-art parallel subgraph counting solution that can be used
for butterfly counting as a special case, we are 349.6--5169x faster.
We find that although the sorting, hashing, and histogramming
aggregation approaches achieve better theoretical complexity, batching
usually performs the best in practice due to lower overheads.

In summary, the contributions of this paper are as follows.

\begin{enumerate}[label=(\textbf{\arabic*})]
\item New parallel algorithms for butterfly counting and peeling.
\item A framework \framework with different ranking and wedge aggregation schemes that can be used for parallel butterfly counting and peeling.
\item Strong theoretical bounds on algorithms obtained using  \framework.
\item A comprehensive experimental evaluation on a 48-core machine demonstrating high parallel scalability and fast running times compared to the best sequential baselines, as well as significant speedups over the state-of-the-art parallel subgraph counting solution. 
\end{enumerate}

The \framework code can be found at \url{https://github.com/jeshi96/parbutterfly}.

%% file: prelims.tex
\section{Preliminaries}\label{sec-par-prim}
\myparagraph{Graph Notation} We take every bipartite graph $G =
(U,V,E)$ to be simple and undirected. For any vertex $v \in U \cup V$,
let $N(v)$ denote the neighborhood of $v$, let $N_2(v)$ denote the
2-hop neighborhood of $v$ (the set of all vertices reachable from $v$
by a path of length $2$), and let $\text{deg}(v)$ denote the degree of
$v$. For added clarity when working with multiple graphs, we let
$N^G(v)$ denote the neighborhood of $v$ in graph $G$ and let
$N^G_2(v)$ denote the 2-hop neighborhood of $v$ in graph $G$. Also, we
use $n = |U|+|V|$ to denote the number of vertices in $G$, and $m =
|E|$ to denote the number of edges in $G$.

A \defn{butterfly} is a set of four vertices $u_1, u_2 \in U$ and
$v_1, v_2 \in V$ with edges $(u_1, v_1)$, $(u_1, v_2)$, $(u_2, v_1)$,
$(u_2, v_2) \in E$. A \defn{wedge} is a set of three vertices $u_1,
u_2 \in U$ and $v \in V$, with edges $(u_1, v), (u_2, v) \in E$. We
call the vertices $u_1, u_2$ \defn{endpoints} and the vertex $v$ the
\defn{center}. Symmetrically, a wedge can also consist of vertices
$v_1, v_2 \in V$ and $u \in U$, with edges $(v_1, u), (v_2, u) \in
E$. We call the vertices $v_1, v_2$ endpoints and the vertex $u$ the
center. We can decompose a butterfly into two wedges that
share the same endpoints but have distinct centers.

The \defn{arboricity} $\alpha$ of a graph is the minimum number of
spanning forests needed to cover the graph. In general, $\alpha$ is
upper bounded by $\BigO{\sqrt{m}}$ and lower bounded by $\BigOmega{1}$
\cite{ChNi85}. Importantly, $\sum_{(u,v) \in E} \text{min}
(\text{deg}(u), \text{deg}(v)) = \BigO{\alpha m}$.

We store our graphs in compressed sparse row (CSR) format, which
requires $\BigO{m+n}$ space. We initially maintain separate offset and
edge arrays for each vertex partition $U$ and $V$, and we assume that
all arrays are stored consecutively in memory.

\myparagraph{Model of Computation} In this paper, we use the work-span
model of parallel computation, with arbitrary forking, to analyze our algorithms.  The
\defn{work} of an algorithm is defined to be the total number of
instructions, and the \defn{span} is defined to be the longest
dependency path~\cite{JaJa92,CLRS}.  We aim for algorithms
to be \defn{work-efficient}, that is, a work complexity that matches
the best-known sequential time complexity. We assume concurrent reads
and writes and atomic adds are supported in $\BigO{1}$ work and span.

\myparagraph{Parallel primitives} We use the following parallel primitives in this paper.
\defn{Prefix sum} takes as
input a sequence $A$ of length $n$, an identity element $\varepsilon$,
and an associative binary operator $\oplus$, and returns the sequence
$B$ of length $n$ where $B[i] = \bigoplus_{j < i} A[j] \oplus
\varepsilon$. 
\defn{Filter} takes as input a sequence $A$ of length
$n$ and a predicate function $f$, and returns the sequence $B$
containing elements $a \in A$ such that $f(a)$ is true, in the same
order that these elements appeared in $A$. Note that filter can be
implemented using prefix sum. Both of these algorithms take $\BigO{n}$
work and $\BigO{\log n}$ span~\cite{JaJa92}.

We also use several parallel primitives in our algorithms for
aggregating equal keys together. \defn{Semisort} groups together equal
keys but makes no guarantee on total order. For a sequence of length
$n$, parallel semisort takes $\BigO{n}$ expected work and $\BigO{\log
  n}$ span with high probability~\cite{GuShSuBl15}.
Additionally, we use parallel hash tables and histograms for
aggregation, which have the same bounds as semisort~\cite{Gil91a,DhBlSh17,
  DhBlSh18,ShBl14}.

%% file: framework.tex
\section{\framework Framework}\label{sec-framework}
In this section, we describe the \framework framework and its
components. Section~\ref{sec-count-framework} describes the procedures
for counting butterflies and Section~\ref{sec-framework-peel}
describes the butterfly peeling procedures.  Section~\ref{sec-alg}
goes into more detail on the parallel algorithms that can be plugged
into the framework, as well as their analysis.

\subsection{Counting Framework}\label{sec-count-framework}

\begin{figure}[!t]
\begin{tcolorbox}
  \textbf{\underline{\framework Framework for Counting}}
  \bigskip
  \begin{enumerate}[label=(\textbf{\arabic*}),topsep=1pt,itemsep=0pt,parsep=0pt,leftmargin=10pt]
  \item  \emph{Rank vertices}: Assign a global ordering, $\textsf{rank}$, to the vertices.
\item \emph{Retrieve wedges}: Retrieve a list $W$ of wedges $(x,y,z)$ where $\rank{y} > \rank{x}$ and $\rank{z} > \rank{x}$.
\item \emph{Count wedges}: For every pair of vertices $(x_1,x_2)$, how many distinct wedges share $x_1$ and $x_2$ as endpoints.
\item \emph{Count butterflies}: Use the wedge counts to obtain the global butterfly count,  per-vertex butterfly counts, or per-edge butterfly counts.
\end{enumerate}
\end{tcolorbox}
\caption{\framework Framework for Counting}\label{fig:framework}
\end{figure}

\begin{figure*}[!t]
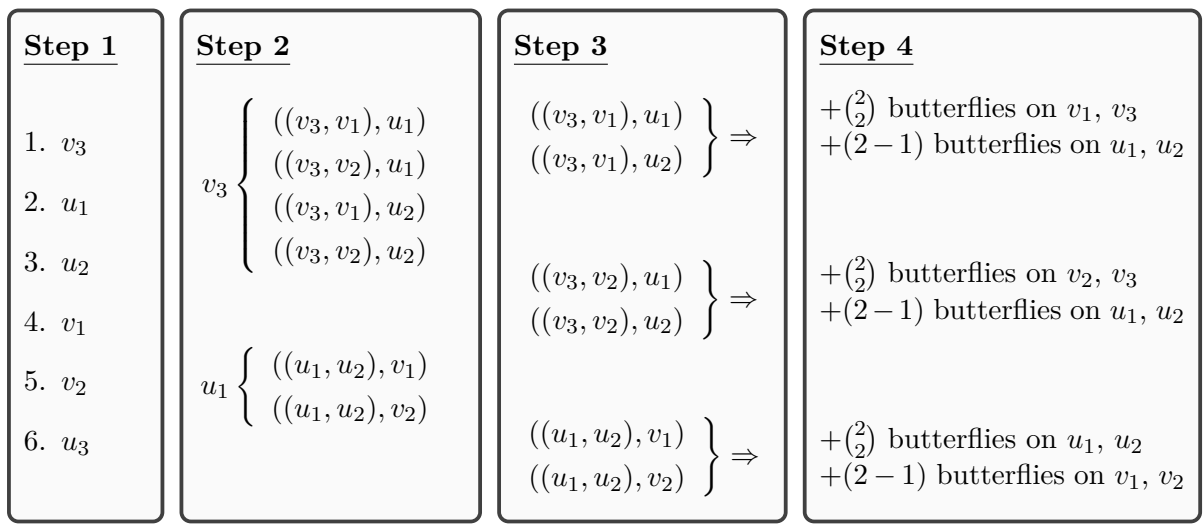

\centering
\begin{tcbraster}[raster columns=4, raster valign=top,
raster force size=false,
raster equal height,
]
\begin{tcolorbox}[width=.13\linewidth, colback=black!2!white,left=2pt,right=2pt]
\textbf{\underline{Step 1}}
\vspace{\baselineskip}
\begin{enumerate}[leftmargin=*]
\item $v_3$
\item $u_1$
\item $u_2$
\item $v_1$
\item $v_2$
\item $u_3$
\end{enumerate}
\end{tcolorbox}
\begin{tcolorbox}[width=.25\linewidth, colback=black!2!white,left=2pt,right=2pt]
\textbf{\underline{Step 2}}
\begin{equation*}
v_3 \begin{cases}\hspace{.2cm}  ((v_3, v_1), u_1) \\ \hspace{.2cm} ((v_3, v_2), u_1) \\ \hspace{.2cm}  ((v_3, v_1), u_2) \\ \hspace{.2cm} ((v_3,v_2),u_2) \end{cases}
\end{equation*}\vspace{\baselineskip}
\begin{equation*}
u_1 \begin{cases}\hspace{.2cm} ((u_1, u_2), v_1) \\ \hspace{.2cm} ((u_1, u_2), v_2) \end{cases}
\end{equation*}
\end{tcolorbox}
\begin{tcolorbox}[width=.24\linewidth, colback=black!2!white,left=2pt,right=2pt]
\textbf{\underline{Step 3}}
\begin{equation*}
\begin{rcases*}
 ((v_3, v_1), u_1) \hspace{.2cm} \\ ((v_3, v_1), u_2) \hspace{.2cm} \end{rcases*}\Rightarrow
\end{equation*}\vspace{1.3\baselineskip}
\begin{equation*}
\begin{rcases*}
 ((v_3, v_2), u_1) \hspace{.2cm} \\ ((v_3,v_2),u_2)\hspace{.2cm}  \end{rcases*}\Rightarrow
\end{equation*}\vspace{1.5\baselineskip}
\begin{equation*}
 \begin{rcases*} ((u_1, u_2), v_1) \hspace{.2cm} \\ ((u_1, u_2), v_2) \hspace{.2cm} \end{rcases*}\Rightarrow
\end{equation*}
\end{tcolorbox}
\begin{tcolorbox}[width=.33\linewidth, colback=black!2!white,left=2pt,right=2pt]
\textbf{\underline{Step 4}}
\vspace{0.6\baselineskip}
\\ $+ {2 \choose 2}$ butterflies on $v_1$, $v_3$ \\
$+ (2-1)$ butterflies on $u_1$, $u_2$
\vspace{1.6\baselineskip}
\\ $+ {2 \choose 2}$ butterflies on $v_2$, $v_3$ \\
$+ (2-1)$ butterflies on $u_1$, $u_2$
\vspace{1.6\baselineskip}
\\ $+ {2 \choose 2}$ butterflies on $u_1$, $u_2$ \\
$+ (2-1)$ butterflies on $v_1$, $v_2$
\end{tcolorbox}
\end{tcbraster}
\caption{We execute butterfly counting per vertex on the graph in Figure~\ref{img_butterfly}. In Step 1, we rank vertices in decreasing order of degree. In Step 2, for each vertex $v$ in order, we retrieve all wedges where $v$ is an endpoint and where the other two vertices have higher rank (the wedges are represented as $((x,z),y)$ where $x$ and $z$ are endpoints and $y$ is the center). In Step 3, we aggregate wedges by their endpoints, and this produces the butterfly counts for Step 4. Note that if we have $w$ wedges that share the same endpoint, this produces ${w \choose 2}$ butterflies for each of the two endpoints and $w-1$ butterflies for each of the centers of the $w$ wedges.
}\label{fig-count-ex}
\end{figure*}

Figure~\ref{fig:framework} shows the high-level structure of the
\framework framework. Step 1 assigns a global ordering to the
vertices, which helps reduce the overall work of the algorithm. Step 2
retrieves all the wedges in the graph, but only where the second and
third vertices of the wedge have higher rank than the first. 
Step 3 counts for every pair of vertices the number of
wedges that share those vertices as endpoints. Step 4 uses the
wedge counts to obtain global, per-vertex, or per-edge butterfly counts.  For each step, there are
several options with respect to implementation, each of which can be
independently chosen and used together. 
Figure~\ref{fig-count-ex} shows an example of executing each of the steps. 
The options within each step of \framework are described in the rest of this section.

\subsubsection{Ranking}\label{sec-framework-retrieval}

The ordering of vertices when we retrieve wedges is significant since
it affects the number of wedges that we process. As we discuss in
Section~\ref{sec-prelim-rank}, Sanei-Mehri \textit{et
  al.}~\cite{SaSaTi18} order all vertices from one bipartition of the
graph first, depending on which bipartition produces the least number
of wedges, giving them practical speedups in their serial
implementation. We refer to this ordering as \defn{side order}.  Chiba
and Nishizeki~\cite{ChNi85} achieve a lower work complexity for
counting by ordering vertices in decreasing order of degree, which we
refer to as \defn{degree order}.

For practical speedups, we also introduce \defn{approximate degree
  order}, which orders vertices in decreasing order of the logarithm
of their degree (\defn{log-degree}). Since the ordering of vertices in many
real-world graphs have good locality, approximate degree order
preserves the locality among vertices with equal log-degree.  We show
in Section \ref{sec-alg-approx} that the work of butterfly counting
using approximate degree order is the same as that of using degree
order.

\defn{Degeneracy order}, also known as the ordering
given by vertex coreness, is a well-studied ordering of vertices given
by repeatedly finding and removing vertices of smallest
degree~\cite{Seidman83, MaBe83}. This ordering can be obtained
serially in linear time using a $k$-core decomposition
algorithm~\cite{MaBe83}, and in parallel in linear
work by repeatedly removing (peeling) all vertices with the smallest
degree from the graph in parallel~\cite{DhBlSh17}. The span of peeling
is proportional to the number of peeling rounds needed to reduce the
graph to an empty graph.
We define \defn{complement degeneracy order} to be the ordering
given by repeatedly removing vertices of
largest degree. This mirrors the idea of decreasing order of
degree, but encapsulates more structural information about the graph.

However, using complement degeneracy order is not 
efficient. The span of finding complement
degeneracy order is limited by the number of rounds needed to reduce a
graph to an empty graph, where each round deletes all maximum degree
vertices of the graph. As such, we define \defn{approximate complement
  degeneracy order}, which repeatedly removes vertices of
largest log-degree. This reduces the number of rounds
needed and closely approximates the number of wedges that
must be processed using complement degeneracy order.
We implement both of these using the parallel bucketing structure of Dhulipala et al.~\cite{DhBlSh17}.

We show
in Section~\ref{sec-alg-approx2} that using complement degeneracy order and approximate
complement degeneracy order give the same work-efficient bounds as
using degree order. 
We show in Section~\ref{sec-imp} that empirically, the
same number or fewer wedges must be processed (compared to both side
and degree order) if we consider vertices in complement degeneracy
order or approximate complement degeneracy order.

In total, the options for ranking are side order, degree order,
approximate degree order, complement degeneracy order, and approximate
complement degeneracy order.

\subsubsection{Wedge aggregation}\label{sec-framework-wcount}

We obtain wedge counts by aggregating wedges by endpoints. \framework
implements fully-parallel methods for aggregation including sorting,
hashing, and histogramming, as well as a partially-parallel batching
method.

We can aggregate the wedges by semisorting key-value pairs where the
key is the two endpoints and the value is the center.  Then, all
elements with the same key are grouped together, and the size of each
group is the number of wedges shared by the two endpoints.  We
implemented this approach using parallel sample sort from the Problem
Based Benchmark Suite (PBBS)~\cite{BlGiSi10, ShBlFiGiKySiTa12} due to
its better cache-efficiency over parallel semisort.

We can also use a parallel hash table to store key-value
pairs where the key is two endpoints and the value is a count. We 
insert the endpoints of all wedges into the table with value
$1$, and sum the values on duplicate keys. The value associated with each key then represents the number of
wedges that the two endpoints share.  We use a parallel hash table
based on linear probing with an atomic addition combining function~\cite{ShBl14}.

Another option is to insert the key-value pairs into a parallel
histogramming structure which counts the number
of occurrences of each distinct key. The parallel histogramming 
structure that we use is implemented using a combination of
semisorting and hashing~\cite{DhBlSh17}.

Finally, in our partially-parallel batching method we process a batch
of vertices in parallel and find the wedges incident on these
vertices. Each vertex aggregates its wedges serially, using an array
large enough to contain all possible second endpoints. The
\defn{simple} setting in our framework fixes the number of vertices in
a batch as a constant based on the space available, while the
\defn{wedge-aware} setting determines the number of vertices
dynamically based on the number of wedges that each vertex processes.

In total, the options for combining wedges are sorting, hashing,
histogramming, simple batching, and wedge-aware batching.

\subsubsection{Butterfly aggregation}
There are two main methods to translate wedge counts into butterfly
counts, per-vertex or per-edge.\footnote{For total counts, butterfly counts can simply be computed and summed in parallel directly.}
One method is to make use of atomic adds, and add the obtained butterfly count for the given
vertex/edge directly into an array, allowing us to obtain butterfly counts without explicit re-aggregation.

The second method is to reuse the aggregation method chosen for the
wedge counting step and use sorting, hashing, or histogramming to
combine the butterfly counts per-vertex or per-edge.\footnote{Note
  that this is not feasible for partially-parallel batching, so in
  that case, the only option is to use atomic adds.}

\subsubsection{Other options}
There are a few other options for butterfly counting in
\framework. First, butterfly counts can be computed per vertex, per
edge, or in total.  For wedge aggregation methods apart from batching,
since the number of wedges can be quadratic in the size of the
original graph, it may not be possible to fit all wedges in memory at
once; a parameter in our framework takes into account the number of
wedges that can be handled in memory and processes subsets of wedges
using the chosen aggregation method until they are all
processed. Similarly, for wedge aggregation by batching, a parameter
takes into account the available space and appropriately determines
the number of vertices per batch.

\framework also implements both edge and colorful sparsification as
described by Sanei-Mehri \textit{et al.}~\cite{SaSaTi18} to obtain
approximate total butterfly counts. For approximate counting, the
sub-sampled graph is simply passed to the framework shown in
Figure~\ref{fig:framework} using any of the aggregation and ranking
choices, and the final result is scaled appropriately.  Note that this
can only be used for total counts.

Finally, Wang \emph{et al.}~\cite{Wang2019} independently describe an
algorithm for butterfly counting using degree ordering, as done in
Chiba and Nishizeki~\cite{ChNi85}, and also propose a cache
optimization for wedge retrieval. Their cache optimization involves
retrieving precisely the wedges given by Chiba and Nishizeki's
algorithm, but instead of retrieving wedges by iterating through the
lower ranked endpoint (for every $v$, retrieve wedges $(v, w,u)$ where
$w, u$ have higher rank than $v$), they retrieve wedges by iterating
through the higher ranked endpoint (for every $u$, retrieve wedges
$(v, w,u)$ where $w, u$ have higher rank than $v$).  Inspired by their
work, we have augmented \framework to include this cache optimization
for all of our orderings.

\subsection{Peeling Framework}\label{sec-framework-peel}
Butterfly peeling classifies induced subgraphs by the number
of butterflies that they contain. Formally, a vertex induced
subgraph is a \defn{k-tip} if it is a maximal induced subgraph such
that for a bipartition, every vertex in that bipartition is
contained in at least $k$ butterflies and every pair of vertices
in that bipartition is connected by a sequence of
butterflies. Similarly, an edge induced subgraph is a \defn{k-wing} if
it is a maximal induced subgraph such that every edge is contained
within at least $k$ butterflies and every pair of edges is connected
by a sequence of butterflies.

The \defn{tip number} of a vertex $v$ is the maximum $k$ such that
there exists a $k$-tip containing $v$, and the \defn{wing number} of
an edge $(u,v)$ is the maximum $k$ such that there exists a $k$-wing
containing $(u,v)$. \defn{Vertex peeling}, or \defn{tip
  decomposition}, involves finding all tip numbers of vertices in
one of the bipartitions, and \defn{edge peeling}, or \defn{wing
  decomposition}, involves finding all wing numbers of edges.

The sequential algorithms for vertex peeling and edge peeling involve
finding butterfly counts and in every round, removing the vertex or
edge contained within the fewest number of butterflies,
respectively. In parallel, instead of removing a single vertex or edge
per round, we remove all vertices or edges that have the minimum
number of butterflies.


The peeling framework is shown in Figure~\ref{fig:peeling-framework},
and supports vertex peeling (tip decomposition) and edge peeling (wing
decomposition). Because it also involves iterating over wedges and
aggregating wedges by endpoint, it contains similar parameters to
those in the counting framework. However, there are a few key
differences between counting and peeling.

\begin{figure}[!t]
\begin{tcolorbox}
  \textbf{\underline{\framework Framework for Peeling}}
  \bigskip
  \begin{enumerate}[label=(\textbf{\arabic*}),topsep=1pt,itemsep=0pt,parsep=0pt,leftmargin=10pt]
  \item  \emph{Obtain butterfly counts}: Obtain per-vertex or per-edge butterfly counts from the counting framework.
  \item \emph{Peel}: Iteratively remove vertices or edges with the lowest butterfly count from the graph until an empty graph is reached.
\end{enumerate}
\end{tcolorbox}
\caption{\framework Framework for Peeling}\label{fig:peeling-framework}
\end{figure}

First, ranking is irrelevant, because all wedges containing a peeled
vertex must be accounted for regardless of order. Also, using atomic add
operations to update butterfly counts is not work-efficient with respect to 
our peeling data structure (see Section \ref{sec-peel-alg}), so we do not have this as an option in our implementation.
Finally, vertex or edge peeling can only be performed if
the counting framework produces per-vertex or per-edge butterfly
counts, respectively.


Thus, the main parameter for the peeling framework is the choice of
method for wedge aggregation: sorting, hashing, histogramming, simple batching, or wedge-aware batching.
These are precisely the same options described in
Section~\ref{sec-framework-wcount}.

%% file: analysis.tex
\section{\framework Algorithms}\label{sec-alg}
We describe here our parallel algorithms for butterfly counting and
peeling in more detail. Our theoretically-efficient parallel
algorithms are based on the work-efficient sequential butterfly
listing algorithm, introduced by Chiba and
Nishizeki~\cite{ChNi85}.
\par Note that Wang \textit{et al.}~\cite{WaFuCh14} proposed the first 
algorithm for butterfly counting, but their algorithm is not work-efficient. They also 
give a simple parallelization of their counting algorithm that is similarly not work-efficient. 
Moreover, Sanei-Mehri \textit{et al.}~\cite{SaSaTi18} and Sariy\"{u}ce and Pinar~\cite{SaPi18} 
give sequential butterfly counting and peeling algorithms respectively, but neither are work-efficient.

\input{counting}

\input{peeling}
\input{approx}

\subsection{Approximate degree ordering}\label{sec-alg-approx}
We show now that using approximate degree ordering in our
preprocessing step also gives work-efficient bounds for butterfly
counting. The proof for this closely follows Chiba and
Nishizeki's~\cite{ChNi85} proof for degree ordering; notably, Chiba
and Nishizeki prove that $\BigO{\sum_{(u,v)\in E} \min
  (\text{deg}(u),\allowbreak \text{deg}(v))} =\allowbreak \BigO{\alpha
  m}$.
\begin{theorem}
Butterfly counting per vertex and per edge using approximate degree
ordering is work-efficient.
\end{theorem}
\begin{proof}
The total work of our counting algorithms, as discussed in
Sections~\ref{sec-count-per-vert} and~\ref{sec-count-per-edge}, is
given precisely by the number of wedges that we must process, or for a
preprocessed graph $G' = (X, E')$, $\BigO{\sum_{x \in X} \sum_{y \in
    N_x(x)} \text{deg}_x(y)}$.

We must have that $\text{deg}_x(y) \leq \text{deg}(y) \leq 2 \cdot
\text{deg}(x)$; otherwise, $y$ would appear before $x$ in approximate
degree order. Moreover, in our double summation, each edge appears
precisely once by virtue of our ordering; thus, our bound becomes
$\BigO{\sum_{(x,y)\in E} (2 \cdot  \min (\text{deg}(x),\allowbreak
  \text{deg}(y)))} =\allowbreak \BigO{\alpha m}$, as desired.
\end{proof}

\subsection{Complement degeneracy ordering}\label{sec-alg-approx2}
We also show that using complement degeneracy ordering and approximate
complement degeneracy ordering in our preprocessing step similarly
gives work-efficient bounds for butterfly counting. As before, the
proof for this closely follows from Chiba and
Nishizeki's~\cite{ChNi85} proof for degree ordering.

\begin{theorem}\label{thm-codegen}
Butterfly counting per vertex and per edge using complement degeneracy 
ordering is work-efficient.
\end{theorem}
\begin{proof}
The total work of our counting algorithms, as discussed in
Sections~\ref{sec-count-per-vert} and~\ref{sec-count-per-edge}, is
given precisely by the number of wedges that we must process, or for a
preprocessed graph $G' = (X, E')$, $\BigO{\sum_{x \in X} \sum_{y \in
    N_x(x)} \text{deg}_x(y)}$.

It is clear that $\text{deg}_x(y) \leq \text{deg}(y)$ by
construction. We would like to show that $\text{deg}_x(y) \leq
\text{deg}(x)$ as well.

Consider the sequential complement $k$-core algorithm. For every round
$r$, let $\text{deg}^r(x)$ denote the degree of $x$ considering only
the induced subgraph on unpeeled vertices. When we peel a vertex $x$
in round $r$, we have for all neighbors $y$ of $x$, $\text{deg}^r(y)
\leq \text{deg}^r(y)$. By our ordering construction, we have
$\text{deg}_x(y) = \text{deg}^r(y)$, and trivially, $\text{deg}^r(x)
\leq \text{deg}(x)$. Thus, $\text{deg}_x(y) \leq \text{deg}(x)$, as
desired.

Thus, the number of wedges that we must process is bounded by
$\BigO{\sum_{x \in X}\allowbreak \sum_{y \in N_x(x)}
  \min(\text{deg}(x),\allowbreak \text{deg}(y))}$. Each edge appears
precisely once in this double summation by virtue of our
ordering. Therefore, our bound becomes $\BigO{\sum_{(x,y)\in E}
  \allowbreak \min (\text{deg}(x),\allowbreak \text{deg}(y))}
=\allowbreak \BigO{\alpha m}$, as desired.
\end{proof}

\begin{theorem}
Butterfly counting per vertex and per edge using approximate
complement degeneracy ordering is work-efficient.
\end{theorem}
\begin{proof}
This follows from the proof of Theorem \ref{thm-codegen}, except that
in each round $r$, when we peel a vertex $u$, we have for all
neighbors $y$ of $x$, $\text{deg}^r(y) \leq 2 \cdot
\text{deg}^r(x)$. Thus, $\text{deg}_x(y) \leq 2 \cdot \text{deg}(x)$,
and so the number of wedges that we must process is bounded by
$\BigO{\sum_{(x,y)\in E} (2 \cdot \min (\text{deg}(x),\allowbreak
  \text{deg}(y)))} =\allowbreak \BigO{\alpha m}$. 
\end{proof}

%% file: counting.tex
\subsection{Preprocessing} \label{sec-prelim-rank}
The main subroutine in butterfly counting involves processing a subset
of wedges of the graph; previous work differ in the way in which they
choose wedges to process.
As mentioned in Section~\ref{sec-framework-retrieval}, Chiba and
Nishizeki~\cite{ChNi85} choose wedges by first ordering vertices by
decreasing order of degree and then for each vertex in order,
extracting all wedges with said vertex as an endpoint and deleting the
processed vertex from the graph. Note that the ordering of vertices
does not affect the correctness of the algorithm -- in fact,
Sanei-Mehri \textit{et al.}~\cite{SaSaTi18} use this precise algorithm
but with all vertices from one bipartition of the graph ordered before
all vertices from the other bipartition. Importantly, Chiba and
Nishizeki's~\cite{ChNi85} original decreasing degree ordering gives
the work-efficient bounds $\BigO{\alpha m}$ on butterfly counting.

Throughout this section, we use decreasing degree ordering to obtain
the same work-efficient bounds in our parallel algorithms. However,
note that using approximate degree ordering, complement degeneracy ordering, and approximate complement degeneracy ordering also gives us these
work-efficient bounds; we prove the work-efficiency of these orderings in Sections \ref{sec-alg-approx} and \ref{sec-alg-approx2}.
Furthermore, our exact and approximate counting algorithms work for any ordering; only the theoretical analysis depends on the ordering.

We use \defn{rank} to denote the index of a vertex in some ordering,
in cases where the ordering that we are using is clear or need not be
specified. We define a modified degree, $\text{deg}_y(x)$, to be the
number of neighbors $z \in N(x)$ such that $\rank{z} > \rank{y}$. We
also define a modified neighborhood, $N_y(x)$, to be the set of
neighbors $z \in N(x)$ such that $\rank{z} > \rank{y}$.

\begin{algorithm}[!t]
  \footnotesize
\caption{Preprocessing}
 \begin{algorithmic}[1]
  \Procedure {PREPROCESS}{$G = (U,V,E), f$}
  \State $X \leftarrow$ \algname{sort}($U\cup V, f$) \Comment{Sort vertices in increasing order of rank according to function $f$}
  \State Let $x$'s rank $R[x]$ be its index in $X$
  \State $E' \leftarrow \{ (R[u], R[v]) \mid (u,v) \in E\}$ \Comment{Rename all vertices to their rank}
  \State $G' = (X, E')$
  \ParFor {$x \in X$}
  \State $N^{G'}(x) \leftarrow$ \algname{sort}($\{ y \mid (x,y) \in E' \}$)
   \Comment{Sort neighbors by decreasing order of rank}
  \State Store $\text{deg}_x(x)$ and $\text{deg}_y(x)$ for all $(x,y) \in E'$
  \EndParFor
  \State \Return $G'$
  \EndProcedure
 \end{algorithmic}
 \label{alg-pre}
 \end{algorithm}

We give a preprocessing algorithm, \algname{preprocess}
(Algorithm~\ref{alg-pre}), which takes as input a bipartite graph and
a ranking function $f$, and renames vertices by their rank in the
ordering.
The output is a general graph (we discard bipartite information in
our renaming). \algname{preprocess} also sorts neighbors by decreasing
order of rank.

\algname{preprocess} begins by sorting vertices in increasing order of
rank.  Assuming that $f$ returns an integer in the range $[0,n-1]$,
which is true in all of the orderings provided in \framework, this can
be done in $\BigO{n}$ expected work and $\BigO{\log n}$ span
w.h.p.\ with parallel integer sort~\cite{RaRe89}. Renaming our graph
based on vertex rank takes $\BigO{m}$ work and $\BigO{1}$ span (to
retrieve the relevant ranks). Finally, sorting the neighbors of our
renamed graph and the modified degrees takes $\BigO{m}$ expected work
and $\BigO{\log m}$ span w.h.p.\ (since these ranks are in the range
$[0,n-1]$).  All steps can be done in linear space. The following
lemma summarizes the complexity of preprocessing.

\begin{lemma}
Preprocessing can be implemented in $\BigO{m}$ expected work,
$\BigO{\log m}$ span w.h.p., and $\BigO{m}$ space.
\end{lemma}

\subsection{Counting algorithms}\label{sec:count}
In this section, we describe and analyze our parallel algorithms for butterfly
counting.

The following equations describe the number of
butterflies per vertex and per edge. Sanei-Mehri \textit{et
  al.}~\cite{SaSaTi18} derived and proved the per-vertex equation, as
based on Wang \textit{et al.}'s~\cite{WaFuCh14} equation for the total
number of butterflies.  We give a short proof of the per-edge
equation.

\begin{lemma}\label{lem-count}
For a bipartite graph $G = (U,V,E)$, the number of butterflies
containing a vertex $u$ is given by
\begin{equation}\label{eq-vert}
\sum_{u' \in N_2(u)} {|N(u) \cap N(u')| \choose 2}.
\end{equation}
The number of butterflies containing an edge $(u,v) \in E$ is given by
\begin{equation}\label{eq-edge}
\sum_{u' \in N(v) \setminus \{u\}} (|N(u) \cap N(u')|-1).
\end{equation}
\end{lemma}
\begin{proof}
The proof for the number of butterflies per vertex is given by
Sanei-Mehri \textit{et al.} \cite{SaSaTi18}. For the number of
butterflies per edge, we note that given an edge $(u,v) \in E$, each
butterfly that $(u,v)$ is contained within has additional vertices $u'
\in U, v' \in V$ and additional edges $(u', v), (u, v'), (u',v') \in
E$. Thus, iterating over all $u' \in N(v)$ (where $u' \neq u$), it
suffices to count the number of vertices $v' \neq v$ such that $v'$ is
adjacent to $u$ and to $u'$. In other words, it suffices to count $v'
\in N(u) \cap N(u') \setminus \{v\}$. This gives us precisely
$\sum_{u' \in N(v) \setminus \{u\}} (|N(u) \cap N(u')|-1)$ as the
number of butterflies containing $(u,v)$.
\end{proof}

Note that in both equations given by Lemma \ref{lem-count}, we iterate
over wedges with endpoints $u$ and $u'$ to obtain our desired counts
(Step 4 of Figure~\ref{fig:framework}).  We now describe how to
retrieve the wedges (Step 2 of Figure~\ref{fig:framework}).

\subsubsection{Wedge retrieval}\label{sec:wedge-retrieval}
\begin{algorithm}[!t]
  \footnotesize
\caption{Parallel wedge retrieval}
 \begin{algorithmic}[1]
  \Procedure {GET-WEDGES}{$G = (V,E), f$}
\ParFor{$x_1 \in V$}
  \ParFor{$i \leftarrow 0$ to $\text{deg}_{x_1}(x_1)$}
   \State $y \leftarrow N(x_1)[i]$ \Comment{$y = i^\text{th}$ neighbor of $x_1$}
     \ParFor{$j \leftarrow 0$ to $\text{deg}_{x_1}(y)$}
      \State $x_2 \leftarrow N(y)[j]$ \Comment{$x_2 = j^\text{th}$ neighbor of $y$}
      \State $f ((x_1, x_2), y)$
    \Comment{$(x_1, x_2)$ are the endpoints, $y$ is the center of the wedge}
     \EndParFor
  \EndParFor
\EndParFor
\State \Return $W$
  \EndProcedure
 \end{algorithmic}
 \label{alg-wedge}
\end{algorithm}

There is a subtle point to make in retrieving all wedges. Once we have
retrieved all wedges with endpoint $x$, Equation~(\ref{eq-vert})
dictates the number of butterflies that $x$ contributes to the second
endpoints of these wedges, and Equation~(\ref{eq-edge}) dictates the
number of butterflies that $x$ contributes to the centers of these
wedges. As such, given the wedges with endpoint $x$, we can
count not only the number of butterflies on $x$, but also the number
of butterflies that $x$ contributes to other vertices of our
graph. Thus, after processing these wedges, there is no need to reconsider wedges containing
$x$ (importantly, there is no need to consider wedges with center
$x$).

From Chiba and Nishizeki's~\cite{ChNi85} work, to minimize the total
number of wedges that we process, we must retrieve all wedges
considering endpoints $x$ in decreasing order of degree, and then
delete said vertex from the graph (i.e., do not consider any other
wedge containing $x$).

We introduce here a parallel wedge retrieval algorithm,
\algname{get-wedges} (Algorithm~\ref{alg-wedge}), that takes
$\BigO{\alpha m}$ work and $\BigO{\log m}$ span. We assume that
\algname{get-wedges} takes as input a preprocessed (ranked) graph and a 
function to apply on each wedge (either for storage or for processing).
The algorithm iterates through all vertices $x_1$ and
retrieves all wedges with endpoint $x_1$ such that the center and second
endpoint both have rank greater than $x_1$ (Lines 3--7).  This is
equivalent to Chiba and Nishizeki's algorithm which deletes vertices
from the graph, but the advantage is that since we do not modify the
graph, all wedges can be processed in parallel.  We process exactly
the set of wedges that Chiba and Nishizeki process, and they prove
that there are $\BigO{\alpha m}$ such wedges.

Since the adjacency lists
are sorted in decreasing order of rank, we can obtain the end index of
the loops on Line 5 using an exponential search in
$\BigO{\text{deg}_{x_1}(y)}$ work and $\BigO{\log
  (\text{deg}_{x_1}(y))}$ span.  Then, iterating over all wedges takes
$\BigO{\alpha m}$ work and $\BigO{1}$ span. In total, we have
$\BigO{\alpha m}$ work and $\BigO{\log m}$ span.

After retrieving wedges, we have to group together the wedges sharing the same endpoints,
and compute the size of each group. We define a subroutine \algname{get-freq} that takes as input 
a function that retrieves wedges, and produces a hash table containing endpoint pairs with their corresponding wedge frequencies. 
This can be implemented using an additive parallel hash table, that increments the count per endpoint pair. Note that parallel 
semisorting or histogramming could also be used, as discussed in
Section~\ref{sec-count-framework}. For an input of length $n$,
\algname{get-freq} takes $\BigO{n}$ expected work and $\BigO{\log n}$
span w.h.p.\ using any of the three aggregation methods as proven in
prior work~\cite{GuShSuBl15,Gil91a,DhBlSh17}. However, semisorting and 
histogramming require $\BigO{n}$ space, while hashing requires space proportional to the
number of unique endpoint pairs. In the case of aggregating all wedges, this is $\BigO{\min(n^2, \alpha m)}$ space.

The following lemma summarizes the complexity of wedge retrieval and aggregation.
\begin{lemma}
Iterating over all wedges can be implemented in $\BigO{\alpha m}$
expected work, $\BigO{\log m}$ span w.h.p., and $\BigO{\min(n^2, \alpha m)}$ space.
\end{lemma}

Note that this is a better worst-case work bound than the work bound of $\BigO{\sum_{v \in V} \text{deg}(v)^2}$ using side order. In the worst-case $\BigO{\alpha m} = \BigO{m^{1.5}}$ while $\BigO{\sum_{v \in V} \text{deg}(v)^2} = \BigO{mn}$.
We have that $mn = \BigOmega{m^{1.5}}$, since $n = \BigOmega{m^{0.5}}$.

\subsubsection{Per vertex}\label{sec-count-per-vert}
\begin{algorithm}[!t]
  \footnotesize
\caption{Parallel work-efficient butterfly counting per vertex}
 \begin{algorithmic}[1]
 \Procedure{COUNT-V-WEDGES}{\algname{get-wedges-func}}
 \State Initialize $B$ to be an additive parallel hash table that stores butterfly counts per vertex
 \State $R\leftarrow$ \algname{get-freq}(\algname{get-wedges-func}) \Comment{Aggregate wedges by wedge endpoints}
 \ParFor{$((x_1, x_2), d)$ in $R$}
 \State Insert $(x_1, {d \choose 2})$ and $(x_2, {d \choose 2})$ in $B$ \Comment{Store butterfly counts per endpoint}
 \EndParFor
 \State $f: ((x_1, x_2), y) \rightarrow $ Insert $(y, R(x_1, x_2) - 1)$ in $B$ 
 \State \Comment{Function to store butterfly counts per center}
 \State \algname{get-wedges-func}($f$) \Comment{Iterate over wedges to store butterfly counts per center}
\State \Return $B$
\EndProcedure
\medskip
  \Procedure {COUNT-V}{$G = (U,V,E)$}
\State $G' = (X,E') \leftarrow$ \algname{preprocess}($G$)
\State \Return \algname{count-v-wedges}(\algname{get-wedges}($G'$))
  \EndProcedure
 \end{algorithmic}
 \label{alg-count-vert}
\end{algorithm}
We now describe the full butterfly counting per vertex algorithm,
which is given as \algname{count-v} in
Algorithm~\ref{alg-count-vert}. As described previously, we implement preprocessing in Line 10, and
we implement wedge retrieval and aggregation by endpoint pairs in Line 3.

We note that following Line 3, by counting the frequency of
wedges by endpoints, for each fixed vertex $x_1$ we have obtained in
$R$ a list of all possible endpoints $(x_1, x_2) \in X \times X$ 
with the size of their intersection $|N(x_1) \cap N(x_2)|$. Thus, by
Lemma~\ref{lem-count}, for each endpoint $x_2$, $x_1$ contributes
${|N(x_1) \cap N(x_2)| \choose 2}$ butterflies, and for each center $y$
(as given in $W$), $x_1$ contributes $|N(x_1) \cap N(x_2)| - 1$
butterflies.
As such, we compute the per-vertex counts by iterating through $R$ to
add the requisite count to each endpoint (Line 5) and iterating through all wedges using \algname{get-wedges} to
add the requisite count to each center (Lines 6--8).

Extracting the butterfly counts from our wedges takes $\BigO{\alpha
  m}$ work (since we are iterating through all wedges) and $\BigO{1}$ span,
and as discussed earlier, \algname{get-freq}
takes $\BigO{\alpha m}$ expected work, $\BigO{\log m}$ span
w.h.p., and $\BigO{\min(n^2, \alpha m)}$ space. The total complexity of butterfly counting per vertex is given
as follows.
\begin{theorem}
Butterfly counting per vertex can be performed in $\BigO{\alpha m}$
expected work, $\BigO{\log m}$ span w.h.p., and $\BigO{\min(n^2, \alpha m)}$ space.
\end{theorem}

\subsubsection{Per edge}\label{sec-count-per-edge}
\begin{algorithm}[!t]
  \footnotesize
\caption{Parallel work-efficient butterfly counting per edge}
 \begin{algorithmic}[1]
  \Procedure{COUNT-E-WEDGES}{\algname{get-wedges-func}}
 \State Initialize $B$ to be an additive parallel hash table that stores butterfly counts per edge  
 \State $R\leftarrow$ \algname{get-freq}(\algname{get-wedges-func}) \Comment{Aggregate wedges by wedge endpoints}
 \State $f: ((x_1, x_2), y) \rightarrow $ Insert $((x_1, y), R(x_1, x_2) - 1)$ and $((x_2, y), R(x_1, x_2) - 1)$ in $B$
 \State \Comment{Function to store butterfly counts per edge}
  \State \algname{get-wedges-func}($f$) \Comment{Iterate over wedges to store butterfly counts per edge}
\State \Return $B$
\EndProcedure
\medskip
  \Procedure {COUNT-E}{$G=(U,V,E)$}
\State $G' = (X,E') \leftarrow$ \algname{preprocess}($G$) 
\State \Return \algname{count-e-wedges}(\algname{get-wedges}($G'$))
  \EndProcedure
 \end{algorithmic}
 \label{alg-count-edge}
\end{algorithm}
We now describe the full butterfly counting per edge algorithm, which
is given as \algname{count-e} in Algorithm~\ref{alg-count-edge}.
We implement preprocessing and wedge
retrieval as described previously, in Line 9 and Line 3,
respectively.

As we discussed in Section~\ref{sec-count-per-vert}, following Step 3
for each fixed vertex $x_1$ we have in $R$ a list of all possible
endpoints $(x_1, x_2) \in X \times X$ with the size of their
intersection $|N(x_1) \cap N(x_2)|$. Thus, by Lemma \ref{lem-count},
we compute per-edge counts by iterating through all of our wedge
counts and adding $|N(x_1) \cap N(x_2)|-1$ to our butterfly counts for
the edges contained in the wedges with endpoints $x_1$ and $x_2$. 
As such, we use \algname{get-wedges} to iterate through all wedges, look up in 
$R$ the corresponding count $|N(x_1) \cap N(x_2)|-1$ on the endpoints, and add 
this count to the corresponding edges.

Extracting the butterfly counts from our wedges takes $\BigO{\alpha
  m}$ work (since we are essentially iterating through $W$) and
$\BigO{1}$ span. 
\algname{get-freq} takes $\BigO{\alpha m}$ expected work, 
$\BigO{\log m}$ span w.h.p., and $\BigO{\min(n^2, \alpha m)}$ space. The total complexity of butterfly counting
per edge is given as follows.
\begin{theorem}
Butterfly counting per edge can be performed in $\BigO{\alpha m}$
expected work, $\BigO{\log m}$ span w.h.p., and $\BigO{\min(n^2, \alpha m)}$ space.
\end{theorem}

%% file: peeling.tex
\subsection{Peeling algorithms}\label{sec-peel-alg}

We now describe and analyze our parallel algorithms for butterfly
peeling. In the analysis, we assume that the relevant per-vertex and per-edge butterfly counts are already given 
by the counting algorithms.
The sequential algorithm for butterfly
peeling~\cite{SaPi18} is precisely the sequential algorithm for
$k$-core~\cite{Seidman83, MaBe83}, except that instead of computing
and updating the number of neighbors removed from each vertex per
round, we compute and update the number of butterflies removed from
each vertex or edge per round. Thus, we base our parallel butterfly
peeling algorithm on the parallel bucketing-based algorithm for
$k$-core in Julienne~\cite{DhBlSh17}. In parallel, our butterfly
peeling algorithm removes (peels) all vertices or edges with the
minimum butterfly count in each round, and repeats until the entire
graph has been peeled.

Zou \cite{Zou16} give a sequential butterfly peeling per edge
algorithm that they claim takes $\BigO{m^2}$ work. However, their algorithm repeatedly scans 
the edge list up to the maximum number of butterflies per edge iterations, so their algorithm 
actually takes $\BigO{m^2 + m \cdot \text{max-b}_e}$ work, where $\text{max-b}_e$ is the 
maximum number of butterflies per edge. This is improved by
Sariy\"{u}ce and Pinar's~\cite{SaPi18} work; Sariy\"{u}ce and Pinar
state that their sequential butterfly peeling algorithms per vertex
and per edge take $\BigO{\sum_{u \in U} \text{deg}(u)^2}$ work and
$\BigO{\sum_{u \in U} \sum_{v_1, v_2 \in N(u)}
  \text{max}(\text{deg}(v_1), \text{deg}(v_2))}$ work, respectively.
They account for the time to update butterfly counts, but do not
discuss how to extract the vertex or edge with the minimum butterfly
count per round. In their implementation, their bucketing structure is
an array of size equal to the number of butterflies, and they
sequentially scan this array to find vertices to peel.  They scan
through empty buckets, and so the time complexity and additional space 
for their butterfly
peeling implementations is on the order of the maximum number of
butterflies per vertex or per edge.

We design a more efficient bucketing structure, which stores non-empty
buckets in a Fibonacci heap~\cite{FrTa84}, keyed by the number of
butterflies. We have an added $\BigO{\log n}$ factor to extract the
bucket containing vertices with the minimum butterfly count. Note that
insertion and updating keys in Fibonacci heaps take $\BigO{1}$
amortized time per key, which does not contribute more to our work. To
use this in our parallel peeling algorithms, we need to ensure that
batch insertions, decrease-keys, and deletions in the Fibonacci are
work-efficient and have low span. We present a parallel Fibonacci heap
and prove its bounds in Section~\ref{sec-fib}.  We show that a batch of $k$
insertions takes $\BigO{k}$ amortized expected work and $\BigO{\log n}$ span
w.h.p., a batch of $k$ decrease-key operations takes $\BigO{k}$
amortized expected work and $\BigO{\log^2 n}$ span w.h.p., and a
parallel delete-min operation takes $\BigO{\log n}$ amortized expected work and
$\BigO{\log n}$ span w.h.p. 

A standard sequential Fibonacci heap gives work-efficient bounds for
sequential butterfly peeling, and our parallel Fibonacci heap gives
work-efficient bounds for parallel butterfly peeling. The work of our
parallel vertex-peeling algorithm 
improves over the sequential algorithm of
Sariy\"{u}ce and Pinar~\cite{SaPi18}.

Our actual implementation uses the bucketing structure from
Julienne~\cite{DhBlSh17}, which is not work-efficient in the context
of butterfly peeling,\footnote{Julienne is work-efficient in the
  context of $k$-core.} but is fast in practice. Julienne materializes
only 128 buckets at a time, and when all of the materialized buckets
become empty, Julienne will materialize the next 128 buckets.  To
avoid processing many empty buckets, we use an optimization to skip
ahead to the next range of 128 non-empty buckets during
materialization.

Finally, we show alternate vertex and edge peeling algorithms that store all of the wedges obtained while counting, which demonstrate 
different work-space tradeoffs.

\begin{algorithm}[t]
  \footnotesize
\caption{Parallel vertex peeling (tip decomposition)}
 \begin{algorithmic}[1]
 \Procedure{GET-V-WEDGES}{$G$, $A$, $f$} \Comment{$f$ is a function used to apply over all wedges without storing said wedges; \algname{count-v-wedges} passes $f$ in when it aggregates wedges for butterfly computations}
 \ParFor{$u_1 \in A$}
  \ParFor{$v \in N(u_1)$ where $v$ has not been previously peeled}
     \ParFor{$u_2 \in N(v)$ where $u_2 \neq u_1$ and $u_2$ has not been previously peeled}
      \State $f((u_1, u_2), v)$
     \EndParFor
  \EndParFor
\EndParFor
\EndProcedure
 \medskip
\Procedure{UPDATE-V}{$G=(U,V,E), B,A$}
\State $B' \leftarrow$ \algname{count-v-wedges}(\algname{get-v-wedges}($G$, $A$))
\State Subtract corresponding counts $B'$ from $B$
\State \Return $B$
\EndProcedure
\medskip
  \Procedure {PEEL-V}{$G=(U,V,E), B$}
  \Comment $B$ is an array of butterfly counts per vertex
\State Let $K$ be a bucketing structure mapping $U$ to buckets based on \# of butterflies
\State $f \leftarrow 0$
\While{$f < |U|$}
\State $A \leftarrow$ all vertices in next bucket (to be peeled)
\State $f \leftarrow f + |A|$
\State $B\leftarrow $\algname{update-v}$(G,B,A)$ \Comment{Update \# butterflies}
\State Update the buckets of changed vertices in $B$
\EndWhile
\State \Return $K$
  \EndProcedure
 \end{algorithmic}
 \label{alg-peel}
\end{algorithm}

\subsubsection{Per vertex}\label{sec-peel-per-vert}
The parallel vertex peeling (tip decomposition) algorithm is given in
\algname{peel-v} (Algorithm~\ref{alg-peel}). Note that 
we peel vertices considering only the bipartition of the graph that 
produces the fewest number of
wedges (considering the vertices in that bipartition as endpoints),
which mirrors Sariy\"{u}ce and Pinar's~\cite{SaPi18} sequential
algorithm and gives us work-efficient bounds for peeling; more
concretely, we consider the bipartition $X$ such that $\sum_{v \in X}
{\text{deg}(v) \choose 2}$ is minimized. Without loss of generality,
let $U$ be this bipartition.

Vertex peeling takes as input the per-vertex butterfly counts from the
\framework counting framework.  We create a bucketing structure
mapping vertices in $U$ to buckets based on their butterfly count
(Line 11). While not all vertices have been peeled, we retrieve the
bucket containing vertices with the lowest butterfly count (Line 16),
peel them from the graph, and compute the wedges removed due
to peeling (Line 16). Finally, we update the buckets of the remaining
vertices whose butterfly count was affected due to the peeling (Line
17).

The main subroutine in \algname{peel-v} is \algname{update-v} (Lines
6--9), which returns a set of vertices whose butterfly counts have
changed after peeling a set of vertices. To compute updated butterfly
counts, we use the equations in Lemma~\ref{lem-count} and precisely
the same overall steps as in our counting algorithms: wedge retrieval,
wedge counting, and butterfly counting. Importantly, in wedge
retrieval, for every peeled vertex $u_1$, we must gather all wedges
with an endpoint $u_1$, to account for all butterflies containing
$u_1$ (from Equation~(\ref{eq-vert})). We process all peeled vertices
$u_1$ in parallel (Line 2), and for each one we find all vertices
$u_2$ in its 2-hop neighborhood, each of which contributes a wedge
(Lines 3--5).  Note that there is a subtle point to make here, where we may 
double-count butterflies if we include wedges containing previously peeled vertices 
or other vertices in the process of being peeled. We can use a separate array to 
mark such vertices, and break ties among vertices in the process of being peeled by 
rank. Then, we ignore these vertices when we iterate over the corresponding wedges in Lines 3 and 4, so these 
vertices are not included when we compute the butterfly contributions of each vertex.

Finally, we aggregate the number of deleted butterflies
per vertex (Line 7), and update the butterfly counts (Line
8).  The wedge aggregation and butterfly counting steps are precisely
as given in our vertex counting algorithm 
(Algorithm~\ref{alg-count-vert}).

The work of \algname{peel-v} is dominated by the total work spent in
the \algname{update-v} subroutine. Since \algname{update-v} will
eventually process in the subsets $A$ all vertices in $U$, the total
work in wedge retrieval is precisely the number of wedges with
endpoints in $U$, or $\BigO{\sum_{u \in U} \text{deg}(u)^2}$. The work
analysis for \algname{count-v-wedges} then follows from a similar
analysis as in Section~\ref{sec-count-per-vert}.  Using our parallel
Fibonacci heap, extracting the next bucket on Line 14 takes
$\BigO{\log n}$ amortized work and updating the buckets on Line 17 is
upper bounded by the number of wedges.

Additionally, the space complexity is bounded above by the space complexity of \algname{count-v-wedges}, which uses a hash table keyed by endpoint pairs. Thus, the space complexity is given by $\BigO{n^2}$.

To analyze the span of \algname{peel-v}, we define $\rho_v$ to be the
\defn{vertex peeling complexity} of the graph, or the number of rounds
needed to completely peel the graph where in each round, all vertices
with the minimum butterfly count are peeled. Then, since the span of
each call of \algname{update-v} is bounded by $\BigO{\log m}$
w.h.p.\ as discussed in Section~\ref{sec-count-per-vert}, and since the
span of updating buckets is bounded by $\BigO{\log^2 m}$ w.h.p., the
overall span of \algname{peel-v} is $\BigO{\rho_v \log^2 m}$ w.h.p.

If the maximum number of per-vertex butterflies  is $\BigOmega{\rho_v
  \log n}$, which is likely true in
practice, then the work of the algorithm described above is faster than Sariy\"{u}ce
and Pinar's~\cite{SaPi18} sequential algorithm, which takes
$\BigO{\text{max-b}_v+\sum_{u \in U}\text{deg}(u)^2}$ work, where
$\text{max-b}_v$ is the maximum number of butterflies per-vertex.

We must now handle the case where $\text{max-b}_v$ is
$\BigO{\rho_v \log n}$. Note that in order to achieve work-efficiency in this case, we must relax the space complexity to $\BigO{n^2 + \text{max-b}_v}$, since we use $\BigO{\text{max-b}_v}$ space to maintain a different bucketing structure. More specifically, 
while we do not know $\rho_v$ at the beginning
of the algorithm, we can start running the algorithm as stated (with the Fibonacci heap), until the number of 
peeling rounds $q$ is equal to $\slfrac{\text{max-b}_v}{\log n}$. If this occurs, then since $q \leq \rho_v$, we have that 
$\text{max-b}_v$ is at most $\rho_v \log n$ (if this does not occur, we know that $\text{max-b}_v$ is greater than $\rho_v
  \log n$, and we finish the algorithm as described above). 
Then, we terminate 
and restart the algorithm using the original bucketing structure of Dhulipala \emph{et
  al.}~\cite{DhBlSh18}, which will give an algorithm with $\BigO{\text{max-b}_v+\sum_{u \in U}\text{deg}(u)^2}$
expected work
and $\BigO{\rho_v \log^2 n}$ span w.h.p. The work bound matches the work bound of
Sariy\"{u}ce and Pinar and therefore, our algorithm is work-efficient.

The overall complexity of butterfly vertex peeling is as follows.
\begin{theorem}
Butterfly vertex peeling can be performed in $\BigO{\min({\normalfont \text{max-b}}_v,\rho_v \log n) + \sum_{u \in U}
  \text{deg}(u)^2}$ expected work, $\BigO{\rho_v \log^2 n}$ span
w.h.p., and $\BigO{n^2 + {\normalfont \text{max-b}}_v}$ space, where ${\normalfont \text{max-b}}_v$ is the maximum number of per-vertex butterflies $\rho_v$ is the vertex peeling complexity. Alternatively, butterfly vertex peeling can be performed in $\BigO{\rho_v \log n + \sum_{u \in U}
  \text{deg}(u)^2}$ expected work, $\BigO{\rho_v \log^2 n}$ span
w.h.p., and $\BigO{n^2}$ space.
\end{theorem}

\begin{algorithm}[t]
  \footnotesize
\caption{Parallel edge peeling (wing decomposition)}
 \begin{algorithmic}[1]
\Procedure{UPDATE-E}{$G=(U,V,E), B,A$}
\State Initialize an additive parallel hash table $B'$ to store updated butterfly counts
\ParFor{$(u_1,v_1) \in A$}
  \ParFor{$u_2 \in N(v_1)$ where $u_2 \neq u_1$ and $(v_1, u_2)$ has not been previously peeled}
  \State $N \leftarrow $ \algname{intersect}$(N(u_1), N(u_2))$, excepting previously peeled edges
  \State Insert $((u_2,v_1), |N| - 1)$ in $B'$ 
    \ParFor{$v_2 \in N$ where $v_2 \neq v_1$}
     \State Insert $((u_1,v_2), 1)$ in $B'$ 
     \State Insert $((u_2,v_2),1)$ in $B'$ 
      \EndParFor
  \EndParFor
\EndParFor
\State Subtract corresponding counts in $B'$ from $B$
\State \Return $B$
\EndProcedure
\medskip
  \Procedure {PEEL-E}{$G=(U,V,E), B$}
    \Comment $B$ is an array of butterfly counts per edge
\State Let $K$ be a bucketing structure mapping $E$ to buckets based on \# of butterflies 
\State $f \leftarrow 0$
\While{$f < m$}
\State $A \leftarrow$ all edges in next bucket (to be peeled)
\State $f \leftarrow f + |A|$
\State $B\leftarrow $\algname{update-e}$(G,B,A)$ \Comment{Update \# butterflies}
\State Update the buckets of changed edges in $B$
\EndWhile
\State \Return $K$
  \EndProcedure
 \end{algorithmic}
 \label{alg-peel-edge}
\end{algorithm}

\subsubsection{Per edge}
While the bucketing structure for butterfly peeling by edge follows
that for butterfly peeling by vertex, the algorithm to update
butterfly counts within each round is different. Based on
Lemma~\ref{lem-count}, in order to obtain all butterflies containing
some edge $(u_1, v_1)$, we must consider all neighbors $u_2 \in N(v_1)
\setminus \{u_1\}$ and then find the intersection $N(u_1) \cap
N(u_2)$. Each vertex $v_2$ in this intersection where $v_2 \neq v_1$
produces a butterfly $(u_1, v_1, u_2, v_2)$. There is no simple
aggregation method using wedges in this scenario; we must find each
butterfly individually in order to count contributions from each
edge. This is precisely the serial update algorithm
that Sariy\"{u}ce and Pinar~\cite{SaPi18} use for edge peeling.

The algorithm for parallel edge peeling is given in \algname{peel-e}
(Algorithm~\ref{alg-peel-edge}).  Edge peeling takes as input the
per-edge butterfly counts from the \framework counting framework.
Line 13 initializes a bucketing structure mapping each edge to a
bucket based on its butterfly count. While not all edges have been
peeled, we retrieve the bucket containing vertices with the lowest
butterfly count (Line 16), peel them from the graph and compute the
wedges that were removed due to peeling (Line 18). Finally, we update
the buckets of the remaining vertices whose butterfly count was
affected due to the peeling (Line 19).

The main subroutine is \algname{update-e} (Lines 1--11), which returns
a set of edges whose butterfly counts have changed after peeling a set
of edges. For each peeled edge $(u_1,v_1)$ in parallel (Line 3), we
find all neighbors $u_2$ of $v_1$ where $u_2\neq u_1$ and compute the
intersection of the neighborhoods of $u_1$ and $u_2$ (Lines 4--5). All
vertices $v_2 \neq v_1$ in their intersection contribute a deleted
wedge, and we indicate the number of deleted wedges on the remaining
edges of the butterfly $(u_2,v_1)$, $(u_1,v_2)$, and $(u_2,v_2)$ in an
array $B'$ (Lines 6--9). As in per-vertex peeling, we can avoid double-counting butterflies 
corresponding to previously peeled edges and edges in the process of being peeled using 
an additional array to mark these edge; we ignore these edges when we iterate over neighbors in Line 4
and perform the intersections in Line 5, so they are not included when we compute the butterfly contributions of each edge.
Finally, we update the butterfly counts (Line
10).

The work of \algname{peel-e} is again dominated by the total work
spent in the \algname{update-e} subroutine.  We can optimize the
intersection on Line 5 by using hash tables to store the adjacency
lists of the vertices, so we only perform
$\BigO{\min(\text{deg}(u),\allowbreak \text{deg}(u'))}$ work when intersecting
$N(u)$ and $N(u')$ (by scanning through the smaller list in parallel and
performing lookups in the larger list).  This gives us
$\BigO{\sum_{(u,v) \in E} \allowbreak \sum_{u' \in N(v)} \allowbreak
  \min(\text{deg}(u),\allowbreak \text{deg}(u'))}$
expected work. 

Additionally, the space complexity is bounded by storing the updated butterfly counts; note that we do not store wedges or aggregate counts per endpoints. Thus, the space complexity is given by $\BigO{m}$. 

As in vertex peeling, to analyze the span of
\algname{peel-e}, we define $\rho_e$ to be the \defn{edge peeling
  complexity} of the graph, or the number of rounds needed to
completely peel the graph where in each round, all edges with the
minimum butterfly count are peeled. The span of
\algname{update-e} is bounded by the span of
updating buckets, giving us $\BigO{\log^2 m}$ span w.h.p. Thus, the
overall span of \algname{peel-e} is $\BigO{\rho_e \log^2 m}$ w.h.p.

Similar to vertex peeling, if the maximum number of per-edge butterflies
 is $\BigOmega{\rho_e \log m}$, which is likely true in practice, then the work
of our algorithm is faster than the sequential algorithm by
Sariy\"{u}ce and Pinar~\cite{SaPi18}. The work of their algorithm is
$\BigO{\text{max-b}_e+ \sum_{(u,v) \in E} \sum_{u' \in N(v)}
  \allowbreak \min(\text{deg}(u), \allowbreak \text{deg}(u'))}$, where
$\text{max-b}_e$ is the maximum number of butterflies per-edge 
(assuming that their intersection is optimized).

To deal with the case where the maximum number of butterflies per-edge
is small, in order to achieve work-efficiency in this case, we must relax the space complexity to $\BigO{m + \text{max-b}_e}$, since we use $\BigO{\text{max-b}_e}$ space to maintain a different bucketing structure. More specifically, 
we can start running the algorithm as stated (with the Fibonacci heap), until the number of 
peeling rounds $q$ is equal to $\slfrac{\text{max-b}_e}{\log m}$. If this occurs, then since $q \leq \rho_e$, we have that 
$\text{max-b}_e$ is at most $\rho_e \log m$ (if this does not occur, we know that $\text{max-b}_e$ is greater than $\rho_e
  \log m$, and we finish the algorithm as described above). 
Then, we terminate 
and restart the algorithm using the original bucketing structure of Dhulipala \emph{et
  al.}~\cite{DhBlSh18}, which will give an algorithm with $\BigO{\text{max-b}_e+\sum_{(u,v) \in E} \sum_{u' \in N(v)}
  \allowbreak \min(\text{deg}(u), \allowbreak \text{deg}(u'))}$
expected work
and $\BigO{\rho_e \log^2 m}$ span w.h.p. Our work bound matches the work bound of
Sariy\"{u}ce and Pinar and therefore, our algorithm is work-efficient.

The overall complexity of butterfly edge peeling is as follows.
\begin{theorem}
Butterfly edge peeling can be performed in $\BigO{\min({\normalfont \mbox{max-b}}_e,\rho_e \log m) +
  \sum_{(u,v) \in E} \sum_{u' \in N(v)} \min(\text{deg}(u),
  \allowbreak \text{deg}(u'))}$ expected work, $\BigO{\rho_e \log^2
  m}$ span w.h.p., and $\BigO{m + {\normalfont \text{max-b}}_e}$ space, where ${\normalfont \text{max-b}}_e$ is the maximum number of per-edge butterflies and $\rho_e$ is the edge peeling complexity. Alternatively, butterfly edge peeling can be performed in $\BigO{\rho_e \log m +
  \sum_{(u,v) \in E} \sum_{u' \in N(v)} \min(\text{deg}(u),
  \allowbreak \text{deg}(u'))}$ expected work, $\BigO{\rho_e \log^2
  m}$ span w.h.p., and $\BigO{m}$ space.
\end{theorem}

\begin{algorithm}[t]
  \footnotesize
\caption{Parallel vertex peeling (tip decomposition)}
 \begin{algorithmic}[1]
\Procedure{WUPDATE-V}{$B,A,W_e, W_c$}
\State Initialize an additive parallel hash table $B'$ to store updated butterfly counts
\ParFor{$x$ in $A$}
  \ParFor{$(y_1, y_2)$ in $W_c(x)$} \Comment{$y_1$ and $y_2$ are endpoints}
       \ParFor{$z$ in $(W_e(y_1))(y_2)$ where $z\neq x$ and $z$ has not been previously peeled} \Comment{$z$ is the center}
       \State Insert $(z, 1)$ in $B'$ \Comment{Update butterfly counts per center}
       \EndParFor
  \EndParFor
  \ParFor{$(y, C)$ in $W_e(x)$ where $y$ has not been previously peeled} 
  \State \Comment{$y$ is the second endpoint, and $C$ is the list of centers}
    \State Insert $(y, {|C| \choose 2})$ in $B'$ \Comment{Update butterfly counts per endpoint}
  \EndParFor
\EndParFor
\State Subtract corresponding counts $B'$ from $B$
\State \Return $B$
\EndProcedure
\medskip
  \Procedure {WPEEL-V}{$G=(U,V,E), B$}
  \Comment $B$ is an array of butterfly counts per vertex
  \State $G' = (X,E') \leftarrow$ \algname{preprocess}($G$)
\State Let $W_e$ be a nested parallel hash table with endpoints as keys and parallel hash tables as values, which are each keyed by the second endpoint and contains lists of centers as values
\State Let $W_c$ be a parallel hash table with centers as keys and lists of endpoint pairs as values
\State $g: ((u_1, u_2), v) \rightarrow$ Insert $(u_1, (u_2, v))$ and $(u_2, (u_1, v))$ in $W_e$, and $(v, (u_1, u_2))$ in $W_c$
\State \algname{get-wedges}($G'$, $g$) \Comment{Store all wedges in $W_e$ and $W_c$}
\State Let $K$ be a bucketing structure mapping $U$ to buckets based on \# of butterflies
\State $f \leftarrow 0$
\While{$f < |U|$}
\State $A \leftarrow$ all vertices in next bucket (to be peeled)
\State $f \leftarrow f + |A|$
\State $B\leftarrow $\algname{wupdate-v}$(B,A,W_e, W_c)$ \Comment{Update \# butterflies}
\State Update the buckets of changed vertices in $B$
\EndWhile
\State \Return $K$
  \EndProcedure
 \end{algorithmic}
 \label{alg-peel-store}
\end{algorithm}

\subsubsection{Per vertex storing all wedges}
We now give an alternate parallel vertex peeling algorithm that stores all wedges. 
The algorithm is given in \algname{wpeel-v} (Algorithm~\ref{alg-peel-store}). As in Section \ref{sec-peel-per-vert}, 
we peel vertices considering only the bipartition of the graph that 
produces the fewest number of
wedges (considering the vertices in that bipartition as endpoints), which we take to be 
$U$, without loss of generality.

The overall framework of \algname{wpeel-v} is similar to that of \algname{peel-v}; we take as input the per-vertex 
butterfly counts from the \framework counting framework, and use a bucketing structure to maintain butterfly counts
and peel vertices (Lines 18--24). The main difference between \algname{wpeel-v} and \algname{peel-v} is in computing updated 
butterfly counts after peeling vertices (Line 23).

\algname{wpeel-v} stores the wedges 
upfront, keyed by each endpoint and by each center, in $W_e$ and $W_c$ respectively (Lines 14--17). Then, in the main subroutine 
\algname{wupdate-v} (Lines 2--11), instead of iterating through the two-hop neighborhood of each vertex $x$ 
to obtain the requisite wedges, we retrieve all wedges containing $x$ by looking $x$ up in $W_e$ and $W_c$. 
Then, to compute updated butterfly counts, we use the equations in Lemma~\ref{lem-count} as before. More precisely, 
we have two scenarios: wedges in which $x$ is a center and wedges in which $x$ is an endpoint. In the first case, we retrieve 
such wedges from $W_c$ (Lines 4--6), and $x$ contributes precisely one butterfly to other wedges that share the same endpoints. 
In the second case, we retrieve such wedges from $W_e$ (Lines 7--9), and Equation~(\ref{eq-vert}) dictates the number of butterflies that $x$ contributes to the 
second endpoint of the wedge.

Note that we do not need to update butterfly counts on any vertex that is not in the same bipartition as $x$. For example, we do not update butterfly counts on $y_1$ and $y_2$ from Line 4 because they are not in the same bipartition as $x$. 

Again, we avoid double-counting butterflies by using a separate array to mark previously peeled vertices and other vertices in the process of being peeled, the details of which are omitted in the pseudocode. We filter these vertices out in Lines 5 and 7. Note that there is no need to check $y_1$ and $y_2$ in Line 4, and the vertices in $C$ in Line 7, because these vertices are not in the same bipartition as $x$ and by construction are not peeled.

The work of \algname{wpeel-v} is dominated by the total work spent in
the \algname{wupdate-v} subroutine. The work of \algname{wupdate-v} over all calls is given by the total number of butterflies $b$. Essentially, iterating over wedges that share the same endpoints in Line 5 amounts to iterating over butterflies, and although we do not remove previously peeled vertices from $W_c$ and $W_e$, this does not affect the work bound because each butterfly can be found at most a constant number of times.

Additionally, the space complexity is given by the space needed to store all wedges in $W_e$ and $W_c$, or $\BigO{\alpha m}$.

Finally, the span of our algorithm follows from the span of \algname{peel-v}. Note that as for \algname{peel-v}, we have two scenarios to consider. In the first scenario, we use no additional space (to the $\BigO{\alpha m}$ space already accounted for), and we use our Fibonacci heap as the bucketing structure, adding $\BigO{\rho_v \log m}$ to our work complexity. In the second scenario, we use the original bucketing structure of Dhulipala \emph{et al.}~\cite{DhBlSh18} (note that there is no need to start with the Fibonacci heap and switch to Dhulipala \emph{et al.}'s bucketing structure here, because of the $\BigO{b}$ in our work complexity). Then, the span improves to $\BigO{\rho_v \log n}$ w.h.p. and we add $\text{max-b}_v$ to the space complexity.

The overall complexity of butterfly vertex peeling is as follows.
\begin{theorem}
Butterfly vertex peeling can be performed in $\BigO{\rho_v \log m + b}$ expected work, $\BigO{\rho_v\log^2 n}$ span
w.h.p., and $\BigO{\alpha m}$ space, where $b$ is the total number of butterflies and $\rho_v$ is the vertex peeling complexity. Alternatively, butterfly vertex peeling can be performed in $\BigO{b}$ expected work, $\BigO{\rho_v \log n}$ span w.h.p., and $\BigO{\alpha m + \allowbreak {\normalfont \text{max-b}}_v}$ space, where ${\normalfont \text{max-b}}_v$ is the maximum number of per-vertex butterflies.
\end{theorem}

\begin{algorithm}[t]
  \footnotesize
\caption{Parallel edge peeling (wing decomposition)}
 \begin{algorithmic}[1]
\Procedure{WUPDATE-E}{$G=(U,V,E), B,A,W_e, W$}
\State Initialize an additive parallel hash table $B'$ to store updated butterfly counts
\ParFor{$(x,y)$ in $A$}
  \ParFor{$z$ in $W(x, y)$ where $(y,z)$ has not been previously peeled} \Comment{$z$ is the second endpoint}
       \ParFor{$w$ in $(W_e(x))(z)$ where $w\neq y$, and $(x,w)$ and $(w,z)$ have not been previously peeled} 
       \State \Comment{$w$ is the center}
       \State Insert $((x,w), 1)$, $((w,z),1)$, and $((y,z),1)$ in $B'$ \Comment{Update butterfly counts per edge}
       \EndParFor
  \EndParFor
  \ParFor{$z$ in $W(y, x)$ where $(x,z)$ has not been previously peeled} \Comment{$z$ is the second endpoint}
       \ParFor{$w$ in $(W_e(y))(z)$ where $w\neq x$, and $(y,w)$ and $(w,z)$ have not been previously peeled} 
       \State \Comment{$w$ is the center}
       \State Insert $((y,w), 1)$, $((w,z),1)$, and $((x,z),1)$ in $B'$ \Comment{Update butterfly counts per edge}
       \EndParFor
  \EndParFor
\EndParFor
\State Subtract corresponding counts $B'$ from $B$
\State \Return $B$
\EndProcedure
\medskip
  \Procedure {WPEEL-E}{$G=(U,V,E), B$}
  \Comment $B$ is an array of butterfly counts per edge
  \State $G' = (X,E') \leftarrow$ \algname{preprocess}($G$)
\State Let $W_e$ be a nested parallel hash table with endpoints as keys and parallel hash tables as values, which are each keyed by the second endpoint and contains lists of centers as values
\State Let $W$ be a parallel hash table with edges as keys and lists of second endpoints as values
\State $g: ((u_1, u_2), v) \rightarrow$ Insert $(u_1, (u_2, v))$ and $(u_2, (u_1, v))$ in $W_e$, and $((u_1,v), u_2)$ and $((u_2,v), u_1)$ in $W$
\State \algname{get-wedges}($G'$, $g$) \Comment{Store all wedges in $W_e$ and $W$}
\State Let $K$ be a bucketing structure mapping $U$ to buckets based on \# of butterflies
\State $f \leftarrow 0$
\While{$f < |U|$}
\State $A \leftarrow$ all edges in next bucket (to be peeled)
\State $f \leftarrow f + |A|$
\State $B\leftarrow $\algname{wupdate-e}$(G,B,A,W_e, W)$ \Comment{Update \# butterflies}
\State Update the buckets of changed edges in $B$
\EndWhile
\State \Return $K$
  \EndProcedure
 \end{algorithmic}
 \label{alg-peel-e-store}
\end{algorithm}

\subsubsection{Per edge storing all wedges}
We now give an alternate parallel edge peeling algorithm that stores all wedges; this algorithm is based on the sequential algorithm given by Wang \textit{et al.} \cite{wang2020efficient}, which takes $\BigO{ b}$ work and $\BigO{\alpha m}$ additional space.
The algorithm is given in \algname{wpeel-e} (Algorithm~\ref{alg-peel-e-store}), and it largely follows \algname{wpeel-v}.

Again, the overall framework of \algname{wpeel-e} is similar to that of \algname{peel-e}, where we take as input the per-edge butterfly
counts from the \framework counting framework, and use a bucketing structure to maintain butterfly counts and peel edges (Lines 20--26).

Like \algname{wpeel-v}, \algname{wpeel-e} stores the wedges 
upfront, keyed by each endpoint in $W_e$ and keyed by the edges in $W$ (Lines 15--19). Then, in the main subroutine 
\algname{wupdate-e} (Lines 2--13), we retrieve all wedges containing each edge $(x,y)$  by looking up the edge in $W$. We first 
consider the case where $x$ is an endpoint and $y$ is a center (Lines 4--7), and we find each butterfly that shares that wedge (Line 5). 
Then, we consider the case where $x$ is a center and $y$ is an endpoint (Lines 8--11), and we repeat this process.

Again, we avoid double-counting butterflies by using a separate array to mark previously peeled edges and other edges in the process of being peeled, the details of which are omitted in the pseudocode. We filter these edges out in Lines 4, 5, 8, and 9.

The work of \algname{wpeel-e} is dominated by the total work spent in
the \algname{wupdate-e} subroutine. The work of \algname{wupdate-e} over all calls is given by the total number of butterflies $b$. Essentially, we find each butterfly individually in Line 7 and 11, and although we do not remove previously peeled edges from $W_e$ and $W$, this does not affect the work bound because each butterfly can be found at most a constant number of times.

Additionally, the space complexity is given by the space needed to store all wedges in $W_e$ and $W$, or $\BigO{\alpha m}$.

Finally, the span of our algorithm follows from the span of \algname{peel-e}. Note that as for \algname{peel-e}, we have two scenarios to consider. In the first scenario, we use no additional space (to the $\BigO{\alpha m}$ space already accounted for), and we use our Fibonacci heap as the bucketing structure, adding $\BigO{\rho_e \log m}$ to our work complexity. In the second scenario, we use the original bucketing structure of Dhulipala \emph{et al.}~\cite{DhBlSh18} (note that there is no need to start with the Fibonacci heap and switch to Dhulipala \emph{et al.}'s bucketing structure here, because of the $\BigO{b}$ in our work complexity). Then, the span improves to $\BigO{\rho_e \log n}$ w.h.p. and we add $\text{max-b}_e$ to the space complexity.

The overall complexity of butterfly edge peeling is as follows.
\begin{theorem}
Butterfly edge peeling can be performed in $\BigO{\rho_e \log m + b}$ expected work, $\BigO{\rho_e\log^2 n}$ span
w.h.p., and $\BigO{\alpha m}$ space, where $b$ is the total number of butterflies and $\rho_e$ is the edge peeling complexity. Alternatively, butterfly edge peeling can be performed in $\BigO{b}$ expected work, $\BigO{\rho_e \log n}$ span w.h.p., and $\BigO{\alpha m + \allowbreak {\normalfont \text{max-b}}_e}$ space, where ${\normalfont \text{max-b}}_e$ is the maximum number of per-edge butterflies.
\end{theorem}

%% file: approx.tex
\subsection{Approximate counting}\label{sec:approx-count}
Sanei-Mehri \textit{et al.}~\cite{SaSaTi18} describe for computing
approximate total butterfly counts based on sampling and graph
sparsification.  Their sparsification methods are shown to have better
performance, and so we focus on parallelizing these methods.  The
methods are based on creating a sparsified graph, running an exact
counting algorithm on the sparsified graph, and scaling up the count
returned to obtain an unbiased estimate of the total butterfly
count. 

The \defn{edge sparsification} method sparsifies the graph by keeping
each edge independently with probability $p$. The butterfly count of
the sparsified graph is divided by $p^4$ to obtain an unbiased
estimate (since each butterfly remains in the sparsified graph with
probability $p^4$). Our parallel algorithm simply applies a filter over
the adjacency lists of the graph, keeping an edge with probability
$p$. This takes $\BigO{m}$ work, $\BigO{\log m}$ span, and $\BigO{m}$ space.

The \defn{colorful sparsification} method sparsifies the graph by
assigning a random color in $[1,\ldots,\lceil 1/p \rceil]$ to each
vertex and keeping an edge if the colors of its two endpoints
match. Sanei-Mehri \textit{et al.}~\cite{SaSaTi18} show that each
butterfly is kept with probability $p^3$, and so the butterfly count
on the sparsified graph is divided by $p^3$ to obtain an unbiased
estimate.  Our parallel algorithm uses a hash function to map each
vertex to a color, and then applies a filter over the adjacency lists
of the graph, keeping an edge if its two endpoints have the same
color. This takes $\BigO{m}$ work, $\BigO{\log m}$ span, and $\BigO{m}$ space.

The variance bounds of our estimates are the same as shown by
Sanei-Mehri \textit{et al.}~\cite{SaSaTi18}, and we refer the reader
to their paper for details.  The expected number of edges in both
methods is $pm$, and by plugging this into the bounds for exact
butterfly counting, and including the cost of sparsification, we
obtain the following theorem.

\begin{theorem}
Approximate butterfly counting with sampling rate $p$ can be performed
in $\BigO{(1+\alpha' p) m}$ expected work and $\BigO{\log m}$ span
w.h.p.,
and $\BigO{\min(n^2,(1+\alpha' p)m}$ space,
where $\alpha'$ is the arboricity of the sparsified graph.
\end{theorem}

%% file: app.tex
\section{Parallel Fibonacci heap}\label{sec-fib}
Fibonacci heaps were first introduced by Fredman and
Tarjan~\cite{FrTa84}. 
In this section, we show that we can parallelize batches of insertion
and decrease-key operations work-efficiently, with logarithmic
span. Also, we show that a single work-efficient parallel delete-min can
be performed with logarithmic span, which is sufficient for our purposes.

Previous work has also explored parallelism in Fibonacci heaps.
Driscoll et al.~\cite{DrGaShTa88} present relaxed heaps, which achieve
the same bounds as Fibonacci heaps, but can be used to obtain a
parallel implementation of Dijkstra's algorithm; however their data
structures do not support batch-parallel insertions or decrease-key
operations.  Huang and Weihl~\cite{Huang1991} and
Bhattarai~\cite{Bhattarai2018} present implementations of parallel
Fibonacci heaps by relaxing the semantics of delete-min, although no
theoretical bounds are given.

A \defn{Fibonacci heap} $H$ consists of heap-ordered trees (maintained
using a root list, which is a doubly linked list), with certain nodes marked and a pointer to the
minimum element. Each node in $H$ is a key-value pair. Let $n$ denote the number of elements in our
Fibonacci heap. The \defn{rank} of a node $x$ is the number of
children that $x$ contains, and the \defn{rank} of a heap is the
maximum rank of any node in the heap. Note that the rank of a
Fibonacci heap is bounded by $\BigO{\log n}$. We also define $t(H)$ to
be the number of trees in $H$, and we define $m(H)$ to be the number
of marked nodes in $H$.

We begin by giving a brief overview of the sequential Fibonacci heap operations:
\begin{itemize}
\item \textbf{Insert} ($\BigO{1}$ work): To insert node $x$, we add $x$ to the root list as a new singleton tree and update the minimum pointer if needed.
\item  \textbf{Delete-min} ($\BigO{\log n}$ amortized work): We delete the minimum node as given by the minimum pointer, and add all of its children to the root list. We update the minimum pointer if needed. We then merge trees until no two trees have the same rank; a merge occurs by taking two trees of the same rank, and assigning the larger root as a child of the smaller root.
\item \textbf{Decrease-key} ($\BigO{1}$ amortized work): To decrease the key of a node $x$, we first check if decreasing the key would violate heap order. If not, we simply decrease the key. Otherwise, we cut the node $x$ and its subtree from its parent, and add it to the root list. If the parent of $x$ was unmarked, we mark the parent. Otherwise, we cut the parent, add it to the root list, and unmark it; we recurse in the same manner on its parent. Finally, we update the minimum pointer if needed
\end{itemize}

The potential function for the amortized analysis is $\Phi(H) = t(H) + 2 \cdot m(H)$.

In our parallel Fibonacci heap, instead of keeping marks on nodes as
boolean values, each node stores an integer number of marks that it
accumulates. Furthermore, instead of maintaining the root list as a
doubly-linked list, we use a parallel hash table~\cite{Gil91a} so that
we can retrieve all roots efficiently in parallel.  Our parallel
Fibonacci heap requires linear space to store, just as in the
sequential version.

\subsection{Batch insertion}\label{sec-batch-insert}
For parallel batch insertion, let $K$ denote the set of key-value
pairs that we are adding to our heap.  Let $k = |K|$, and $n$ be the
size of our heap before insertion. For each key in $K$, we create a
singleton tree.  We resize our parallel hash table if necessary to
make space for the new singleton trees, and then add all new singleton
trees to the root list. Finally, we update the minimum pointer.

Creating new singleton trees takes $\BigO{k}$ work and constant
span. Resizing the hash table takes $\BigO{k}$ amortized expected work
and $\BigO{\log (n+k)}$ span w.h.p.  To update the minimum pointer, we
use a prefix sum between the newly added nodes and the previous
minimum of the heap, which takes $\BigO{k}$ work and $\BigO{\log k}$
span.

The total complexity of batch insertion is given as follows.
\begin{lemma}
Parallel batch insertion of $k$ elements into a Fibonacci heap with $n$ elements takes
$\BigO{k}$ amortized expected work and $\BigO{\log (n+k)}$ span w.h.p.
\end{lemma}

\subsection{Delete-min}

The amortized work of delete-min is $\BigO{\log n}$, but we describe how to parallelize delete-min in order to get a high probability bound for the span.
The parallel delete-min algorithm is given in
Algorithm~\ref{alg-delete}. We first delete the minimum node and add
all of its children to the root list in parallel (Line 2). Then, the
main component of our parallel delete-min operation involves consolidating
trees such that no two trees share the same rank (Lines 3--10).
We place each tree into a group based on its rank
(Line 3), and then merge pairs of trees with the same rank in every
round (Lines 7--9) until there are no longer trees with the same rank.
The final step in our algorithm is updating the minimum pointer using
a prefix sum (Line 11).

\begin{algorithm}[!t]
\footnotesize
  \caption{Parallel delete-min}
 \begin{algorithmic}[1]
\Procedure{PAR-DELETE-MIN}{$H$}
\State Delete the minimum node and add all children to root list
\State Initialize $C$ such that $C[i]$ contains all roots with rank $i$
\While{$\exists$ a group in $C$ with $>1$ root}
  \State Initialize $C'$ to hold updated trees
  \ParFor{$i \leftarrow 0$ to $|C|$}
      \State Partition the roots in $C[i]$ into pairs
      \State If a root is leftover, insert it into $C'[i]$
      \State Merge the trees in every pair and insert the new roots into $C'[i+1]$
  \EndParFor
\State $C \leftarrow C'$
\EndWhile
\State Use prefix sum among the root nodes to update the minimum pointer
\State \Return $H$
\EndProcedure
 \end{algorithmic}
 \label{alg-delete}
\end{algorithm}

After $\BigO{\log n}$ rounds, each group will
necessarily contain at most one tree, which can be shown inductively.
If we assume that after round $i$, all groups $\leq i$ each contain
at most one tree, we see that when we process round $i+1$, no merged
tree can be added to group $\leq i+1$ (since the rank of merged
trees can only increase). Moreover, group $i+1$ contains at most one
leftover root, and all other roots have been merged and inserted into
group $i+2$. Thus, the number of rounds needed to complete our
consolidation step is bounded above by the rank of $H$, which is
$\BigO{\log n}$. 

We use dynamic arrays to represent each group, which allow the insertion
of $x$ elements in $\BigO{x}$ amortized work and $\BigO{1}$ span.
The amortized work of our parallel delete-min operation is asymptotically equal to
the amortized work of the sequential delete-min operation. 
Their actual costs are the same, except for the hash table and dynamic array operations in the parallel version. 
Our actual cost
for merging trees and dynamic array operations 
is
$\BigO{t(H)}$. 
If we let $H'$
represent our heap after performing \algname{par-delete-min}, the
change in potential is $\Delta \Phi \leq \allowbreak t(H') - t(H) \leq
\allowbreak \text{rank}(H') +1 - t(H) = \allowbreak \BigO{\log n -
  t(H)}$, because no two trees have the same rank after performing our
consolidations. Thus, the amortized work for merging trees and dynamic array operations is $\BigO{\log n}$.
The amortized expected work for 
adding the children of the minimum node to the
root list is $\BigO{\log n}$ due to hash table insertions.
Updating the minimum pointer also takes $\BigO{\log n}$ work, and thus the total amortized expected work is $\BigO{\log n}$, as desired.

\begin{algorithm}[!t]
\caption{Parallel batch decrease-key}
\footnotesize
\begin{algorithmic}[1]
\Procedure{BATCH-DECREASE-KEY}{$H$, $K$}
\State \Comment $K$ is an array of triples, holding the key-value pair to be decreased and the updated key
\State Let $M$ be an empty array (to later store marked nodes)
\ParFor{$(k,\_, k') \in K$}
  \If{changing the key to $k'$ violates heap order}
    \State Cut $k$ and add to root list with key $k'$
    \State Add a mark to the original parent of $k$ and add the parent to $M$
  \Else
    \State Change the key to $k'$
  \EndIf
\EndParFor
 \State $M\leftarrow $ nodes in $M$ with $>1$ marks
\While{$M$ is nonempty}
  \State Let $M'$ be an empty array (to later store marked nodes)
  \ParFor{$p \in M$}
    \State Cut $p$ and add to root list
    \State Set $\#$ marks on $p$ to $0$ if $\#$ marks on $p$ is even, and $1$ otherwise
    \State Add a mark to $p$'s original parent and add the parent to $M'$
  \EndParFor
 \State $M\leftarrow$ nodes in $M'$ with $>1$ marks
\EndWhile
\State \Return $H$
\EndProcedure
 \end{algorithmic}
 \label{alg-decrease}
\end{algorithm}

The span of our algorithm is dominated by the span of the while
loop. In particular, note that every iteration of our while loop has
$\BigO{1}$ span, because we can perform the pairwise merges fully in
parallel. As we previously discussed, we have at most $\BigO{\log n}$
iterations of our while loop.
Inserting the children of the minimum node to the
hash table representing the root list takes  $\BigO{\log n}$ span w.h.p.
Therefore, the span of parallel delete-min is
$\BigO{\log n}$ w.h.p. The total complexity of parallel delete-min is as
follows.

\begin{lemma}
Parallel delete-min for a Fibonacci heap takes $\BigO{\log n}$
amortized expected work and $\BigO{\log n}$ span w.h.p.
\end{lemma}

\subsection{Batch decrease-key}
The parallel batch decrease-key operation is given in
\algname{batch-decrease-key} (Algorithm~\ref{alg-decrease}). 

For each decrease key, we check if these decreases
violate heap order (Line 5). If not, we can directly decrease the key (Line 9). Otherwise, we cut these nodes from their
trees and mark their parents (Lines 6--7). Then, we recursively cut
all parents that have been marked more than once, mark their parents,
and repeat (Lines 10--17).

We can maintain the arrays $M$ and $M'$ using a parallel filter in work proportional to the size of our batch and the total number of cuts and $\BigO{\log n}$ span per iteration of the while-loop on Line 11.

On Lines 7 and 16, we record in an array when we would like to mark a
parent. Then, we can semisort the array and use prefix sum to obtain
the number of marks to be added to each parent. This maintains our
work bounds, and has $\BigO{\log n}$ w.h.p. per iteration of the
while-loop on Line 11.

We now focus on the amortized work analysis. Let $k$ denote the
number of keys in $K$ and let $c$ be the total number of cuts that we
perform in this algorithm. Note that decreasing our keys takes
$\BigO{k}$ total work, and the rest of the work is given by the total
number of cuts, or $\BigO{c}$.

Recall that our potential function is $\Phi(H) = t(H) + 2 \cdot
m(H)$. Let $H'$ represent our heap after performing
\algname{batch-decrease-key}. The change in the number of trees is
given by $t(H') - t(H) = c$, since every new cut produces a new tree.

The change in the number of marks is $m(H') - m(H) \leq k-(c-k) =
2k - c$. The argument for this is similar to the sequential
argument.

For each parent node $p$ that is cut, we arbitrarily set a key in $K$
as having \emph{propagated} the cut as follows. Let $c(p)$ denote the
key that propagated the cut to parent $p$, and let $M(p)$ denote the
set of all nodes that marked $p$ in the round immediately before $p$
was cut. Then, we set $c(x) = x$ for each key $x$ in $K$, and we set
$c(p)$ to be an arbitrary key in $C(p) = \{ c(p') \mid p' \in M(p)
\}$; in other words, $c(p)$ is one of the keys that propagated a node
that marked $p$ in the round immediately before $p$ was cut.

Each key $x$ then has a well-defined \emph{propagation path}, which is
the maximal path of nodes that $x$ has propagated. The last node
$\ell$ of the propagation path must either be a root node or a node
whose parent $x$ has not propagated the cut to. Note that $x$ may mark
$\ell$'s parent without cutting this parent from its tree; we call
this mark an \emph{allowance}. In this sense, each key $x$ in $K$ has
one mark in its allowance. We have $k$ marks in the total allowance of
our heap.

It remains, then, to count the change in the number of marks on each propagation path. We claim that if we have already counted the $k$ allowances in the change in the number of marks $(m(H') - m(H) )$, we can now subtract a mark for each cut parent on a propagation path. There are two cases.

If a cut parent $p$ at the start of our algorithm already contained a
mark, then the node that propagated the cut added a mark, canceling
out the previous mark. Thus, $c(p)$ has effectively subtracted a mark
from $p$.

If a cut parent $p$ had no marks at the start of our algorithm, a
child node $p'$ such that $c(p') \neq c(p)$ must have marked $p$
within our algorithm. Necessarily, $c(p')$ must have ended its
propagation path at $p'$, so it charged its mark allowance to
$p$. Then, when the node that propagated the cut, $c(p)$, added a
mark, this cancels out the mark that $p'$ made. Since we have already
counted the mark that $p'$ made in its allowance, we can subtract a
mark to account for the cancellation.

In total, we see that we can subtract a mark for each cut parent on a
propagation path. The number of cut parents is at least $c-k$, so we
have $m(H') - m(H) \leq k - (c-k) = 2k-c$, as desired.

As such, the change in potential is $\Phi(H')-\Phi(H) \leq c + 2(2k-c)
= 4k - c$.  The actual work of \algname{batch-decrease-key} is
$\BigO{k+c}$, and the amortized work is $\BigO{k+c} +
4k-c=\BigO{k}$ by scaling up the units of potential appropriately.

The span of our algorithm is again dominated by the span of the while
loop. We have at most $\BigO{\log n}$ iterations of the while loop,
since it is bounded by the maximum tree height of our heap, which is
bounded by the rank of the heap. Each iteration of the while loop has span $\BigO{\log n}$ w.h.p. Thus, the span of our algorithm is $\BigO{\log^2 n}$ w.h.p.

The overall complexity of parallel
batch decrease-key is as follows.

\begin{lemma}
Parallel batch decrease-key of $k$ elements in a Fibonacci heap with $n$ elements
takes $\BigO{k}$ amortized work and $\BigO{\log^2 n}$ span w.h.p.
\end{lemma}

\subsection{Application to bucketing}
We now discuss more specifically how to apply our batch-parallel
Fibonacci heap to bucketing in butterfly peeling. Each bucket is
represented as a node in the Fibonacci heap, where the key is the
number of butterflies and the value is a parallel hash
table~\cite{Gil91a} containing vertices/edges that contain exactly the
given number of butterflies.  Updates to the bucketing structure
trigger certain operations on the Fibonacci heap.

Throughout the rest of this subsection, we also consider key-value
pairs in the context of bucketing operations.  Bucketing operations
involve processing key-value pairs, where we take the key to be a
given number of butterflies and the value to be a single
vertex/edge. Notably, the value differs from that in the context of a
Fibonacci heap; bucketing operations update single vertices/edges, and
we use the heap to move sets of these vertices/edges to their correct
buckets. To distinguish the key-value pairs in bucketing operations
from the key-value pairs that represent nodes in the Fibonacci heap,
we refer to the latter as \emph{heap key-value pairs}.

Moreover, in order to ensure that all vertices/edges with the same key
are aggregated into a single bucket, we also need a supplemental
parallel hash table~\cite{Gil91a} that stores pointers to buckets,
keyed by the corresponding number of butterflies. The combined
parallel hash table $T$ and batch-parallel Fibonacci heap $H$ form our
bucketing structure $B = (T, H)$. This following analysis assumes $n$
is the number of vertices or edges in the graph, which upper bounds
the size of the Fibonacci heap at any time.  The work bounds for our
operations are amortized, but we can remove the amortization when
summing across all rounds of our peeling algorithms.

\subsubsection{Retrieving the minimum bucket}
To retrieve the minimum bucket, we perform a delete-min operation on
the Fibonacci heap, and given the heap key of the minimum, we remove
the corresponding key in the supplemental hash table $T$. The
delete-min operation on the Fibonacci heap dominates the complexity,
and so retrieving the minimum bucket takes $\BigO{\log
  n}$ amortized expected work and $\BigO{\log n}$ span w.h.p.

\subsubsection{Updating the bucketing structure} \label{sec-batch-insert-bucket}

Updating the bucketing structure involves moving elements to new
buckets based on their updated butterfly counts, which can only
decrease. This involves moving elements between the hash tables of the
buckets in the Fibonacci heap, decreasing the key of some buckets in
the Fibonacci heap, and inserting new buckets into the Fibonacci heap.
The algorithm is shown in Algorithm~\ref{alg-decrease-bucket}.

\begin{algorithm}[!t]
\caption{Bucketing Update Algorithm}
\footnotesize
\begin{algorithmic}[1]
\Procedure{BUCKETING-UPDATE}{$B = (T, H)$, $K$}
\State \Comment $K$ is an array of triples that hold key-value pairs to be decreased and their updated key
\State Let $n_k$ be $\#$ times $k$ appears in a key-value pair in $K$
\State $I = \{\}$ \Comment{Array of key-value pairs to re-insert}
\State $K' = \{\}$ \Comment{Updated $K$}
\ParFor{$(k,v,k') \in K$}
  \If{$n_k=$ size of $k$'s bucket and $v$ is the first element in the bucket}
  \State Add $(k,\{v\},k')$ to $K'$
      \Else
         \State Add $(k',v)$ to $I$ and remove $v$ from bucket $k$
     \EndIf
\EndParFor
\ParFor{$(k,\_,k') \in K'$}
  \State Remove $k$ from $T$ and insert $k'$ into $T$
\EndParFor

\State \algname{batch-decrease-key}($H$, $K'$)
\State  $I' = \{\}$ \Comment{Array of heap key-value pairs to insert into $H$}
\ParFor{each unique $k'$ where $(k',v) \in I$}
\If{$k' \in T$}
\State Add $\{v \mid (k',v) \in I\}$ to the hash table of bucket $k'$
\Else
\State Add $(k', V)$ to $I'$ where $V$ contains all $v$ where $(k',v)\in I$
\EndIf
\EndParFor

\ParFor{$(k',V) \in I'$}
\State Add $k'$ to $T$
\EndParFor
\State \algname{batch-insert}($H$, $I'$)
\State \Return $B$
\EndProcedure
 \end{algorithmic}
 \label{alg-decrease-bucket}
\end{algorithm}

We must first check whether a key-value pair triggers a decrease-key
operation in the Fibonacci heap. If not all values in the bucket need
to be updated, then those values can simply be removed from the bucket and re-inserted
with the updated heap key; the bucket holding the rest of the values can
remain with the original heap key. Otherwise, if all values in the bucket
need to be updated, we keep only the first element in the bucket and
decrease the heap key for that element (Lines 7--8 and 13), and we keep the other values
with their updated keys in an array $I$ to be re-inserted (Lines 9--10). We also also update the supplemental hash table $T$ for the buckets that we decrease the key for (Line 11--12).
We can determine if all values in a bucket need to be
updated by using a semisort to aggregate counts on
the number of times each key appears in $K$. This takes $\BigO{k}$
expected work and $\BigO{\log n}$ span w.h.p.\ where $k=|K|$.

For all key-value pairs that must be reinserted, we first check for
each distinct key whether it appears in our supplemental hash table
$T$. For the heap keys that appear in $T$ (Lines 16--17), we simply add the corresponding
set of values to their existing bucket in the heap as a batch (which
is stored as a hash table, so this consists of performing insertions
to the hash table corresponding to the bucket).  For heap keys that do not
appear in the supplemental hash table $T$, we keep them along with their heap values (a hash table containing the set of
vertices/edges associated with that heap key) in an array $I'$ (Lines 18--19). We also add these new heap keys to $T$ (Lines 20--21).
Then we
perform a batch
insertion in the Fibonacci heap with the set of heap key-value pairs in $I'$ (Line 22).
We can perform the hash
table operations and insertions into the Fibonacci heap in $\BigO{k}$
amortized expected work and $\BigO{\log n}$ span w.h.p.

The batch decrease-key operation on the Fibonacci heap dominates the
complexity, and so updating the bucketing structure for $k$ elements
takes $\BigO{k}$ amortized expected work and $\BigO{\log^2 n}$ span w.h.p.

%% file: eval.tex
\section{Experiments}\label{sec-imp}

\subsection{Environment}
We run our experiments on an m5d.24xlarge AWS EC2 instance, which
consists of 48 cores (with two-way hyper-threading), with 3.1 GHz
Intel Xeon Platinum 8175 processors and 384 GiB of main memory. We use
Cilk Plus's work-stealing scheduler~\cite{BlLe99,Leiserson} and we
compile our programs with g++ (version 7.3.1) using the \texttt{-O3}
flag.
We test our algorithms on a variety of real-world bipartite graphs
from the Koblenz Network Collection (KONECT)~\cite{Kunegis13}.
We
remove self-loops and duplicate
edges from the graph. Table~\ref{table-graphs} describes the properties of these
graphs, including sizes, number of butterflies, and peeling
complexities.

We compare our algorithms against Sanei-Mehri \textit{et al.}'s~\cite{SaSaTi18} and Sariy\"{u}ce and Pinar's~\cite{SaPi18} work, which 
are the state-of-the-art sequential butterfly counting and peeling implementations, respectively.

Notationally, when discussing wedge and butterfly aggregation methods,
we use the prefix ``A'' to refer to using atomic adds for
butterfly aggregation, and we take a lack of prefix to mean that the
wedge aggregation method was used for butterfly
aggregation. ``BatchS'' is the simple version of batching and ``BatchWA'' is the wedge-aware version of batching that
dynamically assigns tasks to workers so they have a roughly equal number
of wedges to process.

\subsection{Results}

\begin{sidewaystable}
  \captionof{table}{These are relevant statistics for the KONECT~\mbox{\cite{Kunegis13}} graphs that we experimented on. Note that we only tested peeling algorithms on graphs for which Sariy\"{u}ce and Pinar's~\mbox{\cite{SaPi18}} serial peeling algorithms completed in less than 5.5 hours. As such, there are certain graphs for which we have no available $\rho_v$ and $\rho_e$ data, and these entries are represented by a dash.}\label{table-graphs}
  \centering
    \small
\begin{tabular}{@{}llllllll@{}}
\toprule
Dataset              & Abbreviation & $|U|$        & $|V|$       & $|E|$         & $\#$ butterflies & $\rho_v$ & $\rho_e$ \\ \midrule
DBLP                 & dblp & 4,000,150  & 1,425,813 & 8,649,016   & 21,040,464   & 4,806 & 1,853 \\
Github               & github &120,867    & 56,519    & 440,237     & 50,894,505      & 3,541 & 14,061\\
Wikipedia edits (it)      &itwiki & 2,225,180  & 137,693   & 12,644,802  & 298,492,670,057      &---&--- \\
Discogs label-style  &discogs& 270,771    & 1,754,823   & 5,302,276  & 3,261,758,502  & 10,676 &   123,859  \\ 
Discogs artist-style & discogs\_style & 383 & 1,617,943 & 5,740,842   & 77,383,418,076  & 374 &   602,142  \\ 
LiveJournal          & livejournal &7,489,073 & 3,201,203 & 112,307,385 & 3,297,158,439,527    &---&--- \\
Wikipedia edits (en)       & enwiki &21,416,395 & 3,819,691 & 122,075,170 & 2,036,443,879,822     &---&---  \\
Delicious user-item   & delicious &33,778,221 & 833,081   & 101,798,957 & 56,892,252,403  &165,850 & ---   \\
Orkut                & orkut&8,730,857 & 2,783,196 & 327,037,487 & 22,131,701,213,295     &---&--- \\
Web trackers & web &27,665,730 & 12,756,244 & 140,613,762 &20,067,567,209,850 & --- & --- 
\end{tabular}

\caption{These are best runtimes in seconds for parallel and sequential butterfly counting from \framework (\textsc{PB}), as well as runtimes from previous work. Note that PGD~\cite{AhmedNRDW17} is parallel, while the rest of the implementations are serial. Also, for the runtimes from our framework, we have noted the ranking used; ${}^*$ refers to side ranking, ${}^\#$ refers to approximate complement degeneracy ranking, and ${}^\circ$ refers to approximate degree ranking. The wedge aggregation method used for the parallel runtimes was simple batching, except the cases labeled with ${}^\lozenge$, which used wedge-aware batching.}\label{table-serial}
\small
      \setlength{\tabcolsep}{1.5pt}

\centering
\begin{tabular}{@{}llllllllllll@{}}
\toprule
 &
\multicolumn{5}{c}{Total Counts}    &
\multicolumn{3}{c}{Per-Vertex Counts}  & \multicolumn{3}{c}{Per-Edge Counts}   \\ 
\cmidrule(lr){2-6}
\cmidrule(lr){7-9}
\cmidrule(lr){10-12}

Dataset              & \shortstack{\textsc{PB} \\ \vspace{\baselineskip}$T_{48\text{h}}$} & \shortstack{\textsc{PB} \\ \vspace{\baselineskip}$T_1$} & \shortstack{Sanei-Mehri \\ et al.~\cite{SaSaTi18} \\ $T_1$}        & \shortstack{PGD~\cite{AhmedNRDW17}  \\ \vspace{\baselineskip} $T_{48\text{h}}$  }  & \shortstack{ESCAPE \\ \cite{Pinar2017} \\ \vspace{\baselineskip}$T_1$}   &\shortstack{\textsc{PB} \\ \vspace{\baselineskip}$T_{48\text{h}}$} & \shortstack{\textsc{PB}\\ \vspace{\baselineskip} $T_1$} & \shortstack{Sariy\"{u}ce and\\ Pinar~\cite{SaPi18} \\ $T_1$ } & \shortstack{\textsc{PB} \\ \vspace{\baselineskip}$T_{48\text{h}}$} &\shortstack{ \textsc{PB} \\ \vspace{\baselineskip} $T_1$}& \shortstack{Sariy\"{u}ce and \\Pinar~\cite{SaPi18} \\ $T_1$}  \\ \midrule
itwiki                 & $0.10^*$ &$1.38^*$ & 1.63&	1798.43&	4.97	&$0.13^*$&$1.43^*$&	6.06&	$0.37^*$&$3.24^\circ$ &	19314.87 \\
discogs              &$0.90^{\# \lozenge}$ &$1.36^\circ$ &4.12 & 234.48 & 2.08 & $0.93^{\# \lozenge}$ &$1.53^\circ$ & 96.09 & $0.59^\circ$&$5.01^*$ & 1089.04\\
livejournal      & $3.83^*$&$35.41^*$ & 37.80 & $>$ 5.5 hrs & 139.06 &$5.65^*$& $36.22^*$ & 158.79 & $10.26^\circ$&$105.65^*$ & $>$ 5.5 hrs\\
enwiki  & $8.29^\circ$&$68.73^*$& 69.10 & $>$ 5.5 hrs & 151.63 &$11.75^\circ$& $75.10^*$ & 608.53 & $16.73^\circ$&$167.69^*$ & $>$ 5.5 hrs \\ 
delicious &$13.52^\circ$ & $165.03^*$& 162.00 & $>$ 5.5 hrs & 286.86 & $18.36^\circ$ &$182.00^*$ & 1027.12 & $23.58^\circ$&$321.02^\circ$ & $>$ 5.5 hrs\\ 
orkut         &$35.07^{*}$&$423.02^*$ &403.46 & $>$ 5.5 hrs & 1321.20 & $66.19^{*}$&$439.02^*$ & 2841.27 & $131.07^{* \lozenge}$&$1256.83^*$ & $>$ 5.5 hrs\\
web       & $12.18^\circ$& $115.53^\circ$&4340 & $>$ 5.5 hrs & 172.77 & $15.89^\circ$&$195.43^\circ$ & $>$ 5.5 hrs & $17.40^\#$&$218.15^\circ$ & $>$ 5.5 hrs
\end{tabular}
\end{sidewaystable}

\begin{figure*}[t]
  \centering
\includegraphics[width=\textwidth, page=1]{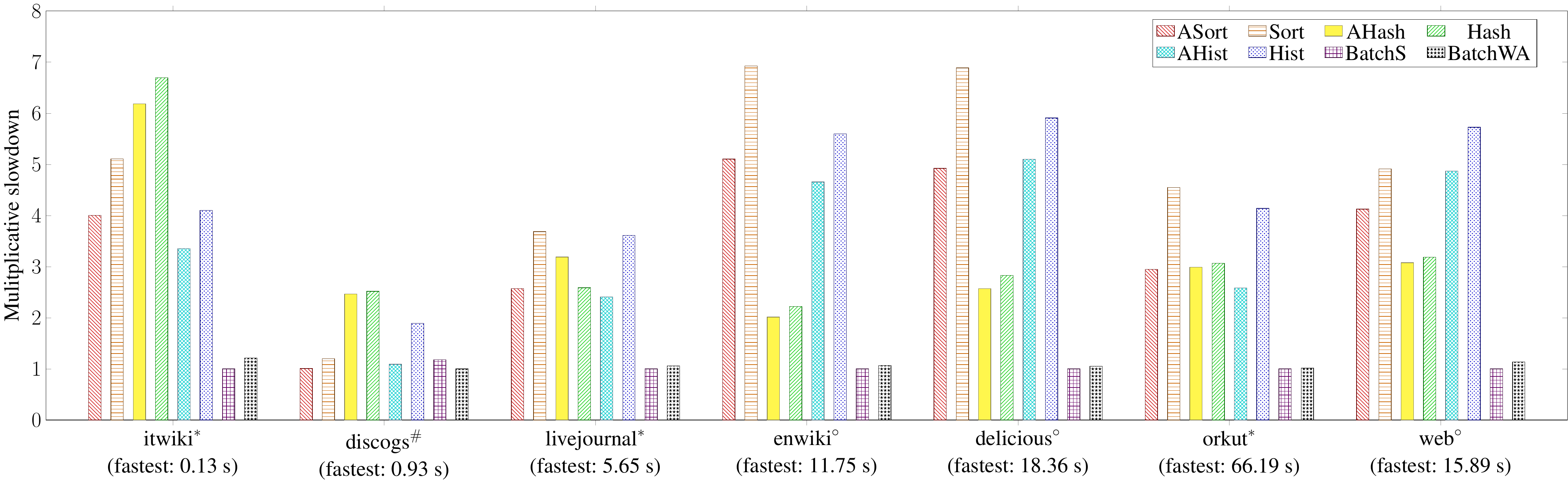}
  \caption{These are the parallel runtimes for butterfly counting per vertex, considering different wedge aggregation and butterfly aggregation methods. We consider the ranking that produces the fastest runtime for each graph; ${}^*$ refers to side ranking, ${}^\#$ refers to approximate complement degeneracy ranking, and ${}^\circ$ refers to approximate degree ranking. All times are scaled by the fastest parallel time, as indicated in parentheses.}\label{fig-count-vert}
\end{figure*}

\subsubsection{Butterfly counting}
Figures~\ref{fig-count-vert},~\ref{fig-count-edge},
and~\ref{fig-count-tot} show the runtimes over different aggregation
methods for counting per vertex, per edge, and in total, respectively,
for the seven datasets in Table~\ref{table-graphs} with sequential
counting times exceeding 1 second.  The times are normalized to the
fastest combination of aggregation and ranking methods for each
dataset.  We find that simple batching and wedge-aware batching give
the best runtimes for butterfly counting in general. Among the
work-efficient aggregation methods, hashing and histogramming with
atomic adds are often faster than sorting, particularly for larger
graphs due to increased parallelism and locality, respectively.  Our
fastest parallel runtimes for each dataset for total, per-vertex, and
per-edge counts are shown in Table~\ref{table-serial}.

We also implemented sequential algorithms for butterfly counting in
\framework that do not incur any of the parallelism overheads.
Table~\ref{table-serial} includes the runtimes for our sequential
counting implementations, as well as runtimes for implementations from
previous works, all of which we tested on the same machine.  The code
from Sanei-Mehri \emph{et al.} and Sariy\"{u}ce and
Pinar~\cite{SaPi18} are serial implementations for global and local
butterfly counting, respectively.  PGD~\cite{AhmedNRDW17} is a
parallel framework for counting subgraphs of up to size 4 and ESCAPE
is a serial framework for counting subgraphs of up to size 5. We timed
only the portion of the codes that counted butterflies.  Our
configurations achieve parallel speedups between 6.3--13.6x over the
best sequential implementations for large enough graphs.\footnote{By
  ``large enough,'' we mean graphs for which the sequential counting
  algorithms take more than 2 seconds to complete.} We also improve
upon PGD by 349.6--5169x due to
having a work-efficient algorithm.

Figures \ref{fig-self-vert} and \ref{fig-self-edge} show our
self-relative speedups on livejournal for per-vertex and per-edge counting,
respectively. Across all rankings, on
livejournal, we achieve self-relative speedups between 10.4--30.9x for
per-vertex counting, between 9.2--38.5x for per-edge counting, and
between 7.1--38.4x for in total counting.

\begin{figure*}[t] \centering
  \includegraphics[width=\textwidth, page=2]{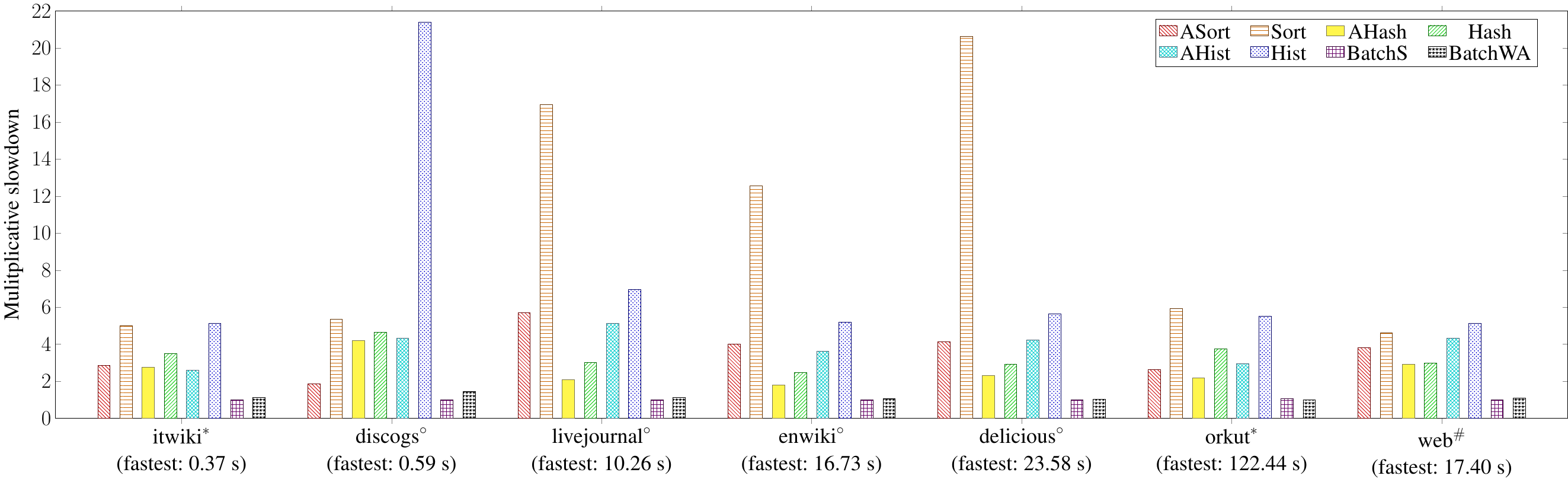}
  \caption{These are the parallel runtimes for butterfly counting per edge, considering different wedge aggregation and butterfly aggregation methods. We consider the ranking that produces the fastest runtime for each graph; ${}^*$ refers to side ranking, ${}^\#$ refers to approximate complement degeneracy ranking, and ${}^\circ$ refers to approximate degree ranking. All times are scaled by the fastest parallel time, as indicated in parentheses.}\label{fig-count-edge}
\end{figure*}
\begin{figure*}[t]
  \centering
\includegraphics[width=\textwidth, page=10]{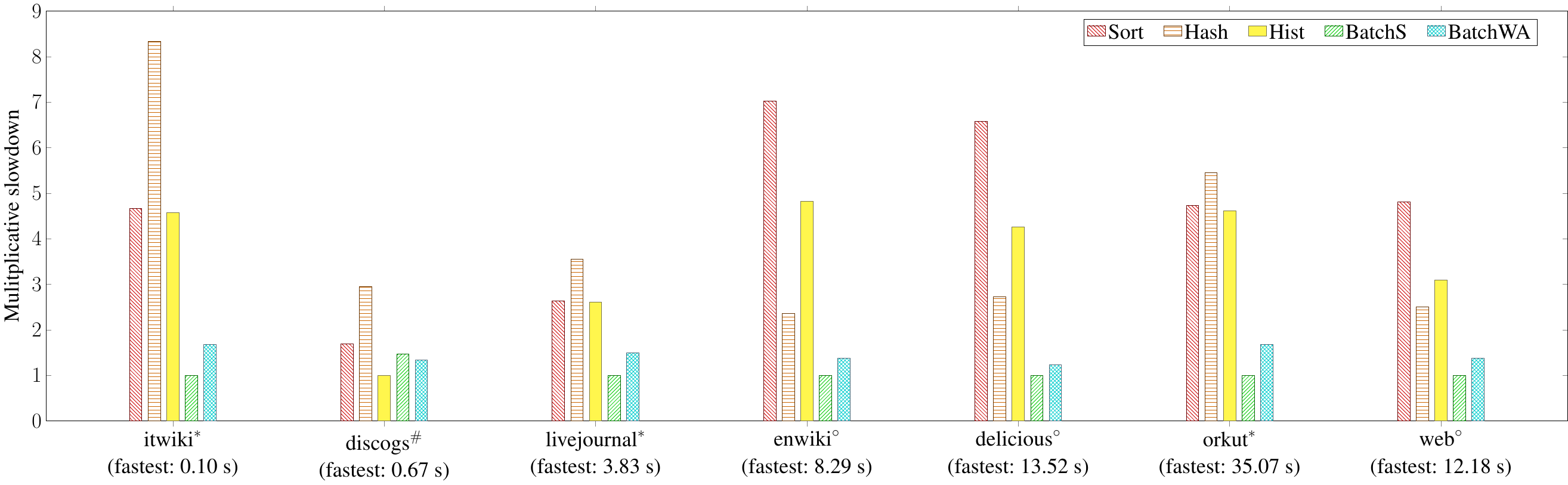}
  \caption{These are the parallel runtimes for butterfly counting in total, considering different wedge aggregation methods (butterfly aggregation does not apply). We consider the ranking that produces the fastest runtime for each graph; ${}^*$ refers to side ranking, ${}^\#$ refers to approximate complement degeneracy ranking, and ${}^\circ$ refers to approximate degree ranking. All times are scaled by the fastest parallel time, as parentheses.}\label{fig-count-tot}
\end{figure*}

\begin{figure*}[!t]
 \begin{minipage}{.44\textwidth}
 \vspace{\baselineskip}
\centering
\includegraphics[width=\textwidth, page=8]{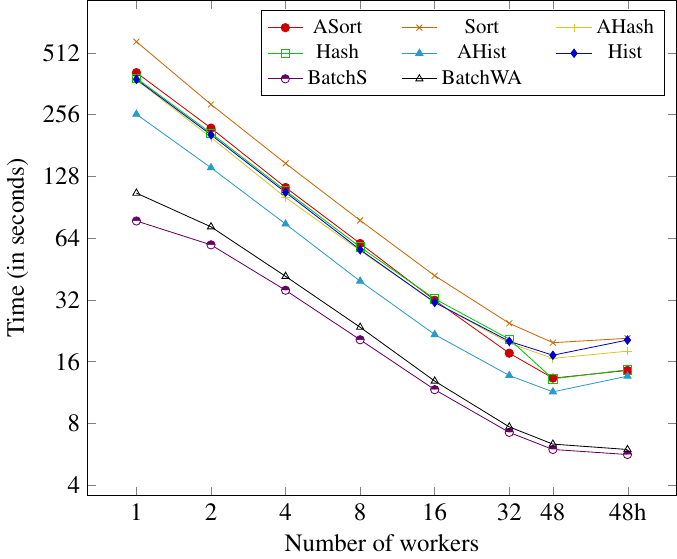}
  \caption{These are the runtimes for butterfly counting per vertex on livejournal using side ranking, over different numbers of threads. The self-relative speedups are between 13.7--28.3x.}\label{fig-self-vert}
   \end{minipage}\hspace{1cm}
 \begin{minipage}{.44\textwidth}
   \vspace{\baselineskip}
  \centering
\includegraphics[width=\textwidth, page=7]{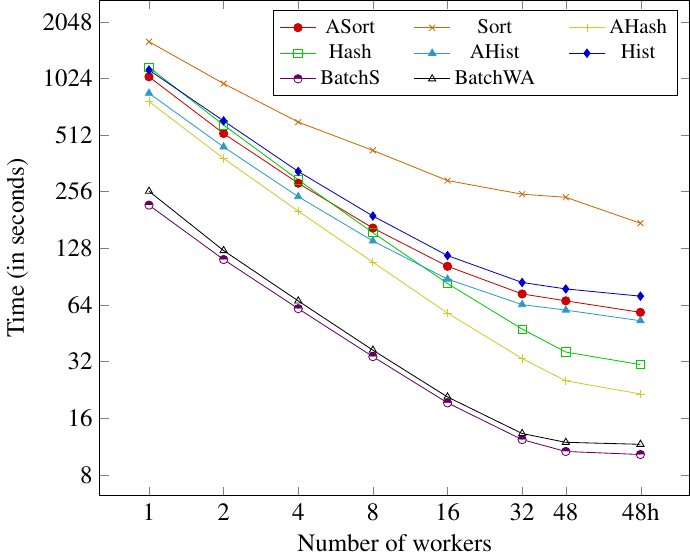}
  \caption{These are the runtimes for butterfly counting per edge on livejournal using approximate degree ranking, over different numbers of threads. The self-relative speedups are between 15.9--38.0x.}\label{fig-self-edge}
   \end{minipage}
\end{figure*}

\begin{figure*}[!t]
  \centering
\includegraphics[width=\textwidth, page=6]{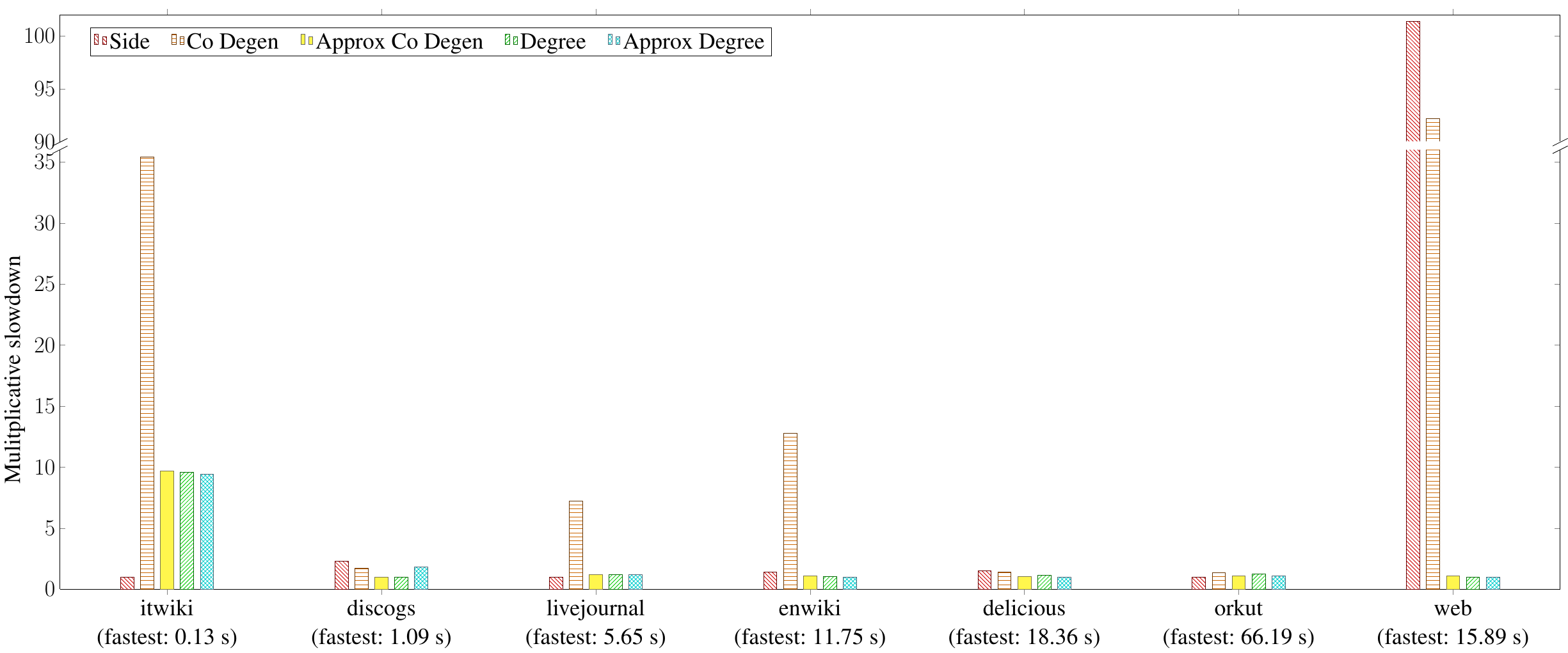}
  \caption{These are the runtimes for butterfly counting per vertex, considering different rankings. We use simple batching as our wedge aggregation method. All times are scaled by the fastest runtime, as indicated in parentheses. Moreover, the time taken to rank each graph is included in the runtimes.}\label{fig-rank}
\end{figure*}

\subsubsection{Ranking}
Figure \ref{fig-rank} shows the runtimes for butterfly counting per
vertex for different rankings using the simple batching method. The
times are normalized to the time for the fastest ranking for each dataset.  Side
ordering outperforms the other rankings for itwiki, livejournal, and
orkut, while approximate complement degeneracy, approximate degree,
and degree orderings outperform side ordering for discogs, enwiki,
delicious, and web.

Note that different rankings change the number of wedges that we must
process; in particular, we found that complement degeneracy and approximate complement degeneracy minimizes
the number of wedges that we process across all of the real-world
graphs considered. However, complement degeneracy is not a feasible
ordering in practice, since the time for ranking often exceeds the time for the actual
counting. Moreover, side ordering often outperforms the other rankings
due to better locality, especially if the number of wedges processed
by the other rankings does not greatly exceed the number of wedges
given by side ordering. We found that the approximate complement degeneracy, degree, and approximate degree orders perform similarly, and these orderings are all efficient to compute.

As such, we devise a metric $f$ that in general determines
whether side ordering outperforms other rankings. If we let $w_s$ be
the number of wedges processed by using side ordering and $w_r$ be the
number of wedges processed by using another ranking, our metric is
$\slfrac{(w_s - w_r)}{w_s}$. If this metric is below $0.1$, then side
ordering will outperform or perform just as well as other
rankings. Table~\ref{table-ranks} shows this metric
across all of the rankings  in \framework. Note that this
metric is fairly similar across these other rankings, and is
particularly high for web, explaining the significant speedup obtained
by using approximate degree ordering over side ordering for web.
The $f$ metric can be computed at runtime, before actually ranking the graph,
to decide which ordering to use.

\subsubsection{Approximate counting}
Figure \ref{fig-approx} shows runtimes for both colorful sparsification and edge sparsification on orkut, as well as the corresponding single-threaded times. We see that over a variety of probabilities $p$ we achieve self-relative speedups between 4.9--21.4x.

\begin{table}[!t]
\caption{These are the fractional values $f = \slfrac{(w_s - w_r)}{w_s}$, where $w_s$ is the number of wedges that must be processed using side ordering and $w_r$ is the number of wedges that must be processed using the labeled ordering. Note that for itwiki, the number of wedges produced by degree ordering is precisely equal to the number of wedges produced by side ordering.}\label{table-ranks}
  \small
\centering
\begin{tabular}{@{}lcccc@{}}
\toprule

Dataset   & Complement Degeneracy & \shortstack{Approx\\Complement Degeneracy}      & Degree     & \shortstack{Approx\\Degree}     \\ \midrule
itwiki                 &  0.021 & 0.020 & 0 & -0.00033 \\
discogs              &0.97 & 0.97 & 0.97 & 0.96\\
livejournal      & 0.035 & 0.033 & 0.011 & -0.019 \\
enwiki  & 0.47 & 0.47 & 0.45 & 0.46\\ 
delicious &  0.60 & 0.60 & 0.59 & 0.59\\ 
orkut         & 0.070 & 0.064 & 0.042 & 0.029\\
web       &  0.95 & 0.95 & 0.95 & 0.95
\end{tabular}
\end{table}

\begin{figure*}[!t]
  \centering
\includegraphics[width=.4\textwidth, page=11]{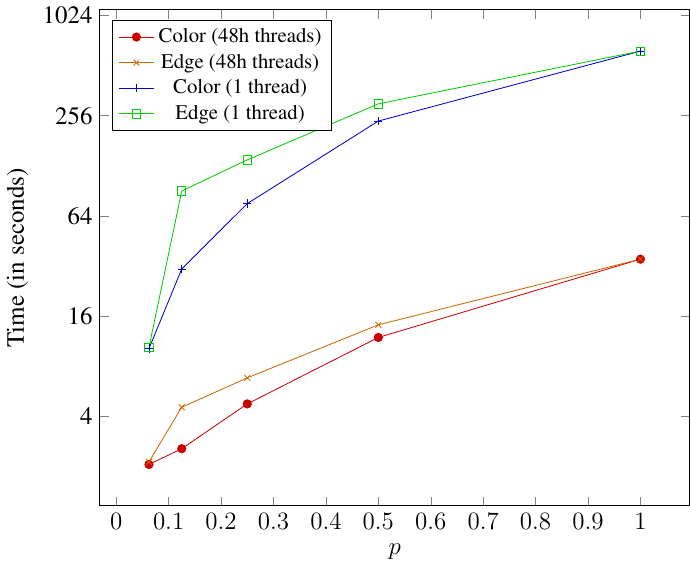}
  \caption{These are the runtimes for colorful sparsification and edge sparsification over different probabilities $p$. We considered both the runtimes on 48 cores hyperthreaded and on a single thread. We ran these algorithms on orkut, using simple batch aggregation and side ranking. }\label{fig-approx}
\end{figure*}

\begin{figure*}[!t]
  \centering
\includegraphics[width=\textwidth, page=5]{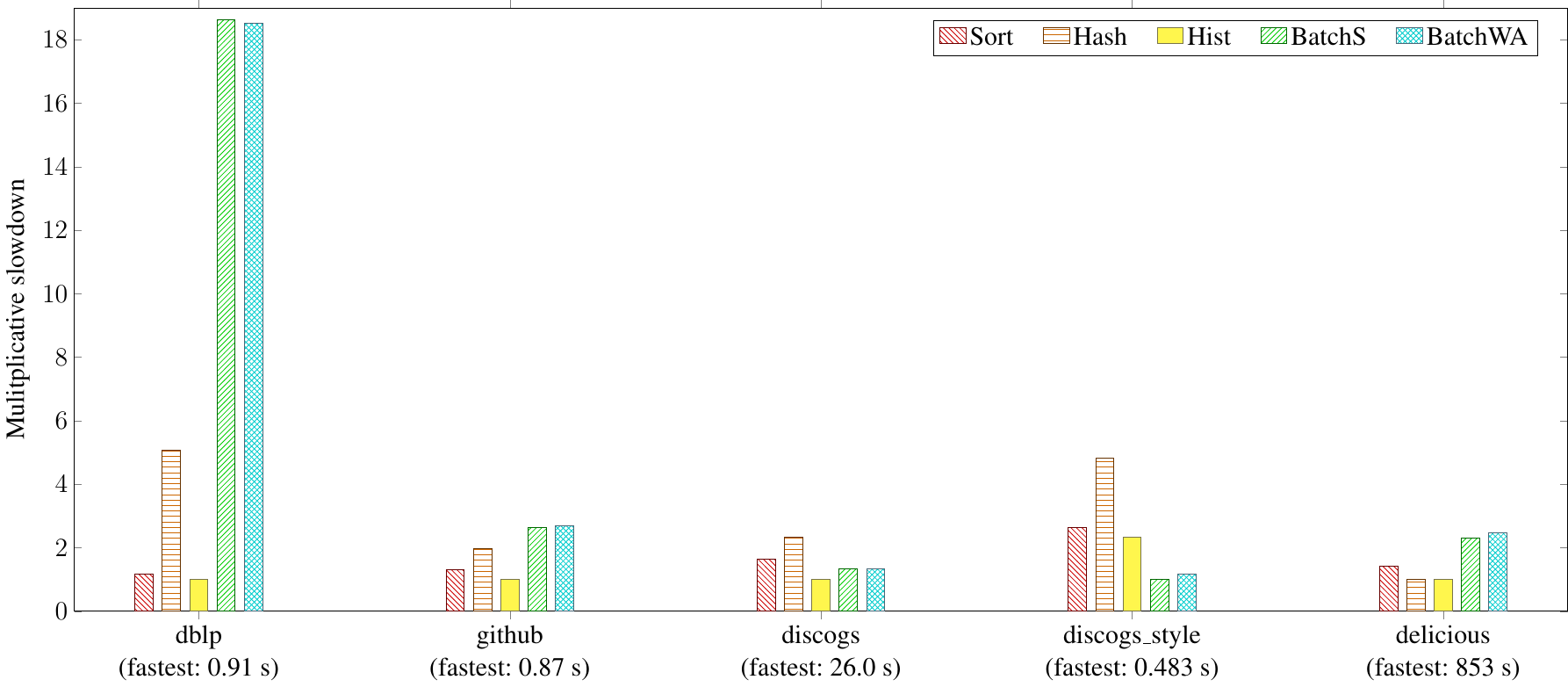}
  \caption{These are the parallel runtimes for butterfly vertex peeling with different wedge aggregation methods (these runtimes do not include the time taken to count butterflies). All times are scaled by the fastest parallel time, as indicated in parentheses. Also, note that the runtimes for discogs\_style represent single-threaded runtimes; this is because we did not see any parallel speedups for discogs\_style, due to the small number of vertices that were peeled.}\label{fig-peel-vert}
\end{figure*}
\begin{figure*}[!t]
  \centering
\includegraphics[width=.8\textwidth, page=9]{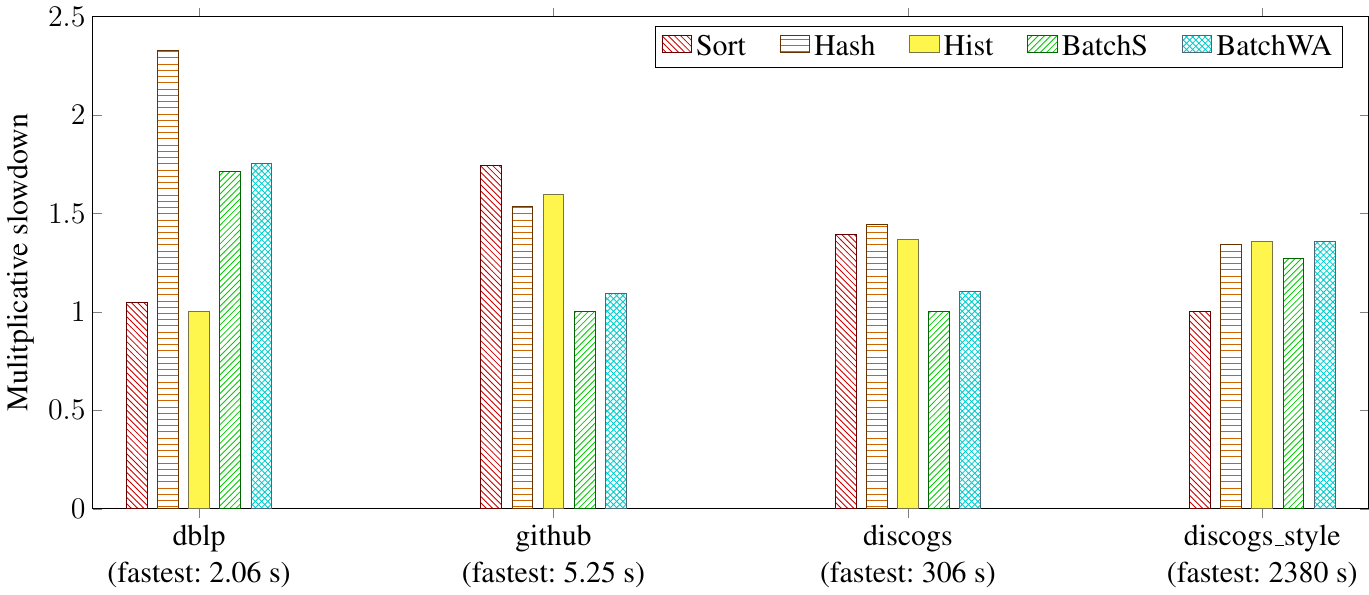}
  \caption{These are the parallel runtimes for butterfly edge peeling with different wedge aggregation methods (these runtimes do not include the time taken to count butterflies). All times are scaled by the fastest parallel time, as indicated in parentheses.}\label{fig-peel-edge}
\end{figure*}

\begin{table}[!t]
  \caption{These are runtimes in seconds for parallel and single-threaded butterfly peeling from \framework (\textsc{PB}) and serial butterfly peeling from Sariy\"{u}ce and Pinar~\cite{SaPi18}. Note that these runtimes do not include the time taken to count butterflies. For the runtimes from \framework, we have noted the aggregation method used; ${}^*$ refers to simple batching, ${}^\#$ refers to sorting and ${}^\circ$ refers to histogramming.}\label{table-peel-serial}
\small
\centering
\begin{tabular}{@{}lllllll@{}}
\toprule
&
\multicolumn{3}{c}{Vertex Peeling}  & \multicolumn{3}{c}{Edge Peeling}   \\ 
\cmidrule(lr){2-4}
\cmidrule(lr){5-7}
Dataset    & \shortstack{\textsc{PB} \\ \vspace{\baselineskip}$T_{48\text{h}}$} & \shortstack{\textsc{PB} \\ \vspace{\baselineskip}$T_{1}$} & \shortstack{Sariy\"{u}ce and\\ Pinar~\cite{SaPi18} \\ $T_1$} & \shortstack{\textsc{PB} \\ \vspace{\baselineskip}$T_{48\text{h}}$} & \shortstack{\textsc{PB} \\ \vspace{\baselineskip}$T_{1}$}  & \shortstack{Sariy\"{u}ce and\\ Pinar~\cite{SaPi18}  \\ $T_1$} \\ \midrule
dblp        &   $0.91^\circ$&$2.40^{\circ}$ &2.06 & $2.06^\circ$&$16.90^{\circ}$ &6.93 \\
github       &  $0.87^\circ$& $1.03^{\circ}$ & 1.15 & $5.25^*$ &$18.00^*$ &18.82 \\
discogs    & $26^\circ$&$53.10^*$ &157.14 & $306^*$&$2160^*$ &2149.54 \\
discogs\_style &$0.48^*$& $0.48^*$ &14826.16 & $2380^\#$& $15600^*$ &16449.56 \\ 
delicious &  $853^\circ$&$1900^*$ & 2184.27 &--- &--- &---
\end{tabular}
\end{table}

\subsubsection{Butterfly peeling}
Figures~\ref{fig-peel-vert} and~\ref{fig-peel-edge} show the runtimes
over different wedge aggregation methods for vertex peeling and edge
peeling, respectively (the runtimes do not include the time for counting butterflies). We only report times for the datasets for which
finished within 5.5 hours.  We find that for vertex peeling,
aggregation by histogramming largely gives the best runtimes, while
for edge peeling, all of our aggregation methods give similar results.

We compare our parallel peeling times to our single-threaded peeling times and serial peeling times from
Sariy\"{u}ce and Pinar's~\cite{SaPi18} implementation, which we ran in
our environment and which are shown in Table~\ref{table-peel-serial}.
Compared to Sariy\"{u}ce and Pinar~\cite{SaPi18}, we achieve speedups
between 1.3--30696x for vertex peeling and between 3.4--7.0x for edge
peeling. Our speedups are highly variable because they depend heavily
on the peeling complexities and the number of empty buckets processed.
Our largest speedup of 30696x occurs for vertex peeling on
discogs\_style where we are able to efficiently skip over many empty buckets,
while the implementation of Sariy\"{u}ce and Pinar sequentially iterates over the empty buckets.

Moreover, comparing our parallel peeling times to their corresponding
single-threaded times, we achieve speedups between 1.0--10.7x for
vertex peeling and between 2.3--10.4x for edge peeling.
We did not
see self-relative parallel speedups for vertex peeling on
discogs\_style, because the total number of vertices peeled (383) was too
small.

\input{figures}

%% file: figures.tex
\begin{figure*}[!t]
  \centering
\includegraphics[width=\textwidth, page=1]{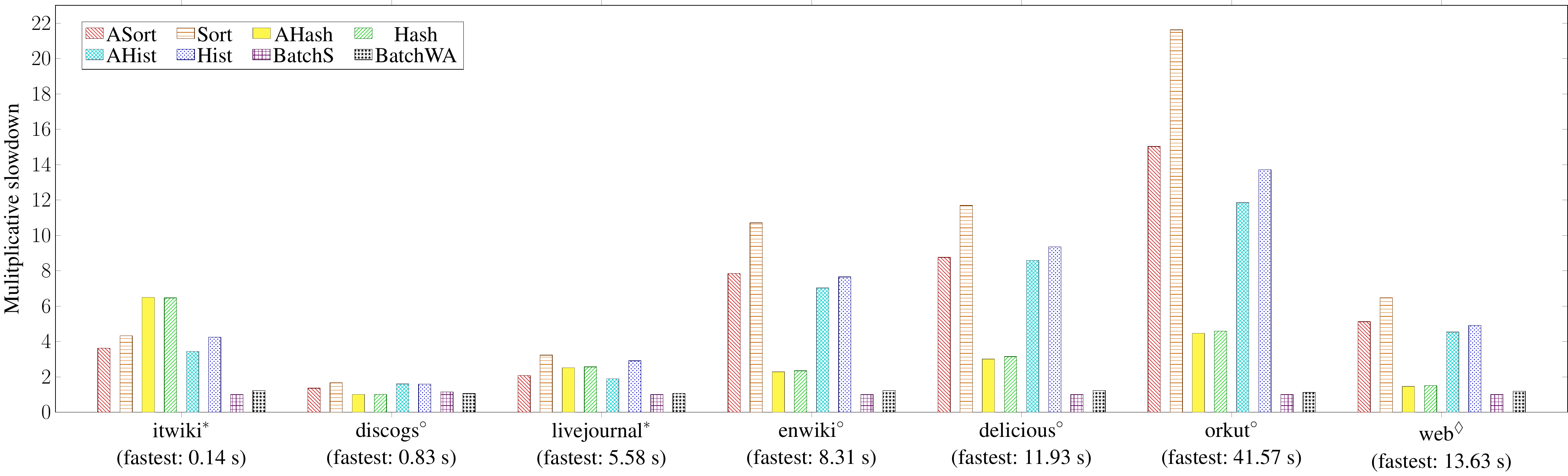}
  \caption{These are the parallel runtimes for butterfly counting per vertex (using the cache optimization), considering different wedge aggregation and butterfly aggregation methods. We consider the ranking that produces the fastest runtime for each graph; ${}^*$ refers to side ranking, ${}^\lozenge$ refers to degree ranking, and ${}^\circ$ refers to approximate degree ranking. All times are scaled by the fastest parallel time, as indicated in parentheses.}\label{inv-fig-count-vert}
\end{figure*}

\subsection{Cache optimization for butterfly counting}

We find that using Wang \emph{et al.}'s~\cite{Wang2019} cache optimization for total, per-vertex, and per-edge parallel butterfly counting gives speedups of up to 1.7x of our parallel butterfly counting algorithms without the cache optimization, considering the best aggregation and ranking methods for each case. 
We present the butterfly counting experiments with the cache optimization enabled in this section.
The trends are similar to the results without
the cache optimization enabled.
Note that the cache optimization does not always improve the performance of 
butterfly counting; for certain graphs, the best butterfly counting time is obtained 
without using the optimization.

\subsection{Butterfly counting}
\begin{sidewaystable}[ph!]
  \caption{These are runtimes in seconds for sequential butterfly counting from \framework (using the cache optimization), as well as runtimes from previous work. Note that PGD~\cite{AhmedNRDW17} is parallel, while the rest of the implementations are serial. Also, for the runtimes from our framework, we have noted the ranking used; ${}^*$ refers to side ranking, ${}^\#$ refers to approximate complement degeneracy ranking, ${}^\lozenge$ refers to degree ranking, and ${}^\circ$ refers to approximate degree ranking. 
We have also noted the wedge aggregation and butterfly aggregation methods used for the parallel runtimes; ${}^\star$ refers to hashing with atomic adds, ${}^\blacklozenge$ refers to histogramming for both aggregation methods, and ${}^\square$ refers to wedge-aware batching. The rest of the parallel runtimes were obtained using simple batching.}\label{inv-table-serial}
  \small

\centering
\begin{tabular}{@{}llllllllllll@{}}
\toprule
 &
\multicolumn{5}{c}{Total Counts}    &
\multicolumn{3}{c}{Per-Vertex Counts}  & \multicolumn{3}{c}{Per-Edge Counts}   \\ 
\cmidrule(lr){2-6}
\cmidrule(lr){7-9}
\cmidrule(lr){10-12}

Dataset              & \shortstack{\textsc{PB} \\ \vspace{\baselineskip}$T_{48\text{h}}$} & \shortstack{\textsc{PB} \\ \vspace{\baselineskip}$T_1$}  & \shortstack{Sanei-Mehri \\ et al.~\cite{SaSaTi18} \\ $T_1$}        & \shortstack{PGD~\cite{AhmedNRDW17}  \\ \vspace{\baselineskip} $T_{48\text{h}}$  }  & \shortstack{ESCAPE~\cite{Pinar2017} \\ \vspace{\baselineskip}$T_1$}   &\shortstack{\textsc{PB} \\ \vspace{\baselineskip}$T_{48\text{h}}$} & \shortstack{\textsc{PB} \\ \vspace{\baselineskip}$T_1$} & \shortstack{Sariy\"{u}ce and\\ Pinar~\cite{SaPi18} \\ $T_1$ } & \shortstack{\textsc{PB} \\ \vspace{\baselineskip}$T_{48\text{h}}$} & \shortstack{\textsc{PB} \\ \vspace{\baselineskip}$T_1$} & \shortstack{Sariy\"{u}ce and \\Pinar~\cite{SaPi18} \\ $T_1$}  \\ \midrule
itwiki                 & $0.10^{*}$& $1.22^\circ$  & 1.63&	1798.43&	4.97	&$0.14^*$&$1.46^*$&	6.06&	$0.41^{*\square}$&$4.77^*$ &	19314.87 \\
discogs              & $0.81^{\# \blacklozenge}$&$1.20^\lozenge$ &4.12 & 234.48 & 2.08 &$0.83^{\circ \star}$& $0.90^\#$ & 96.09 & $0.52^\circ$&$2.97^\circ$ & 1089.04\\
livejournal      &$3.83^{*}$& $35.12^*$ & 37.80 & $>$ 5.5 hrs & 139.06 &$5.58^*$& $38.91^*$ & 158.79 &$8.91^\lozenge$& $106.07^*$ & $>$ 5.5 hrs\\
enwiki  & $5.59^\circ$&$59.18^\circ$& 69.10 & $>$ 5.5 hrs & 151.63 & $8.31^\circ$&$54.71^\circ$ & 608.53 & $11.58^\circ$&$152.92^\circ$ & $>$ 5.5 hrs \\ 
delicious &$8.05^\circ$&  $82.43^\circ$& 162.00 & $>$ 5.5 hrs & 286.86 &$11.93^\circ$& $96.18^\circ$& 1027.12 & $15.98^\circ$&$210.70^\circ$  & $>$ 5.5 hrs\\ 
orkut        &$22.47^\#$ &$414.02^*$ &403.46 & $>$ 5.5 hrs & 1321.20 &$41.57^\circ$& $557.10^*$ & 2841.27 & $89.52^\circ$&$1141.05^*$ & $>$ 5.5 hrs\\
web       & $8.08^\circ$& $61.68^\circ$&4340 & $>$ 5.5 hrs & 172.77 & $13.63^\lozenge$&$77.82^\circ$ & $>$ 5.5 hrs &$14.08^\lozenge$ &$154.35^\circ$ & $>$ 5.5 hrs
\end{tabular}
\end{sidewaystable}
Figures~\ref{inv-fig-count-vert},~\ref{inv-fig-count-edge},
and~\ref{inv-fig-count-tot} show the runtimes over different
aggregation methods for counting per vertex, per edge, and in total,
respectively, for the seven datasets in Figure~\ref{table-graphs} with
sequential counting times exceeding 1 second.  The times are
normalized to the fastest combination of aggregation and ranking
methods for each dataset.  Considering different wedge and butterfly
aggregation methods, we again find that simple batching and
wedge-aware batching give the best runtimes for butterfly counting in
general. Of the work-efficient aggregation methods, hashing with
atomic adds is often the fastest, particularly for larger graphs due
to increased parallelism.

\begin{figure*}[!t] \centering
  \includegraphics[width=\textwidth, page=2]{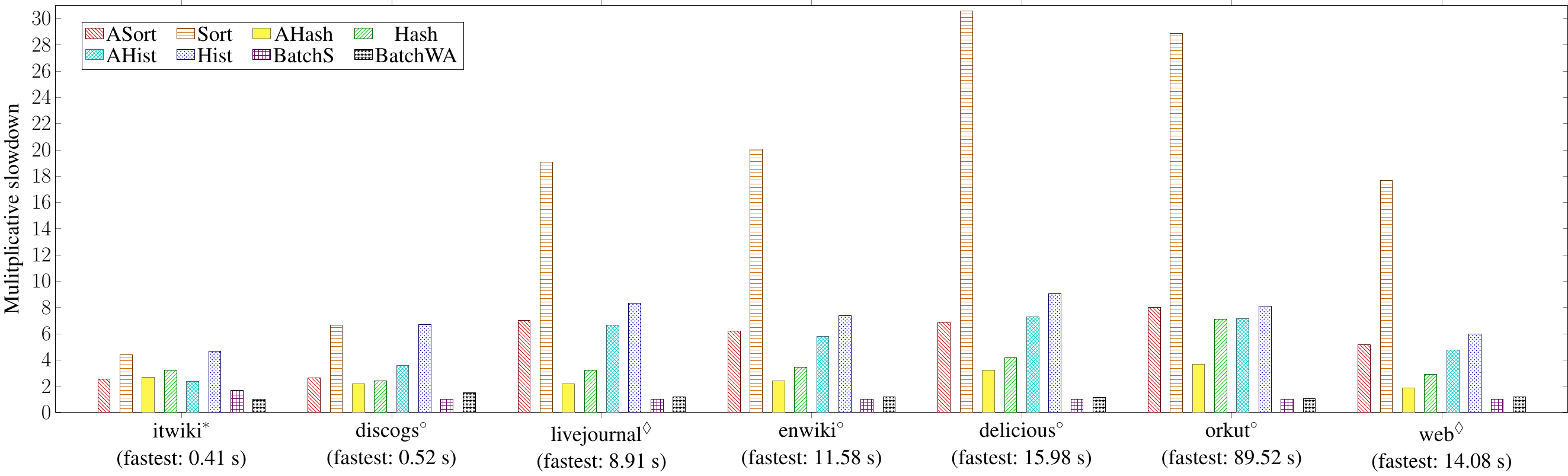}
  \caption{These are the parallel runtimes for butterfly counting per edge (using the cache optimization), considering different wedge aggregation and butterfly aggregation methods. We consider the ranking that produces the fastest runtime for each graph; ${}^*$ refers to side ranking, ${}^\lozenge$ refers to degree ranking, and ${}^\circ$ refers to approximate degree ranking. All times are scaled by the fastest parallel time, as indicated in parentheses.}\label{inv-fig-count-edge}
\end{figure*}
\begin{figure*}[!t]
  \centering
\includegraphics[width=\textwidth, page=10]{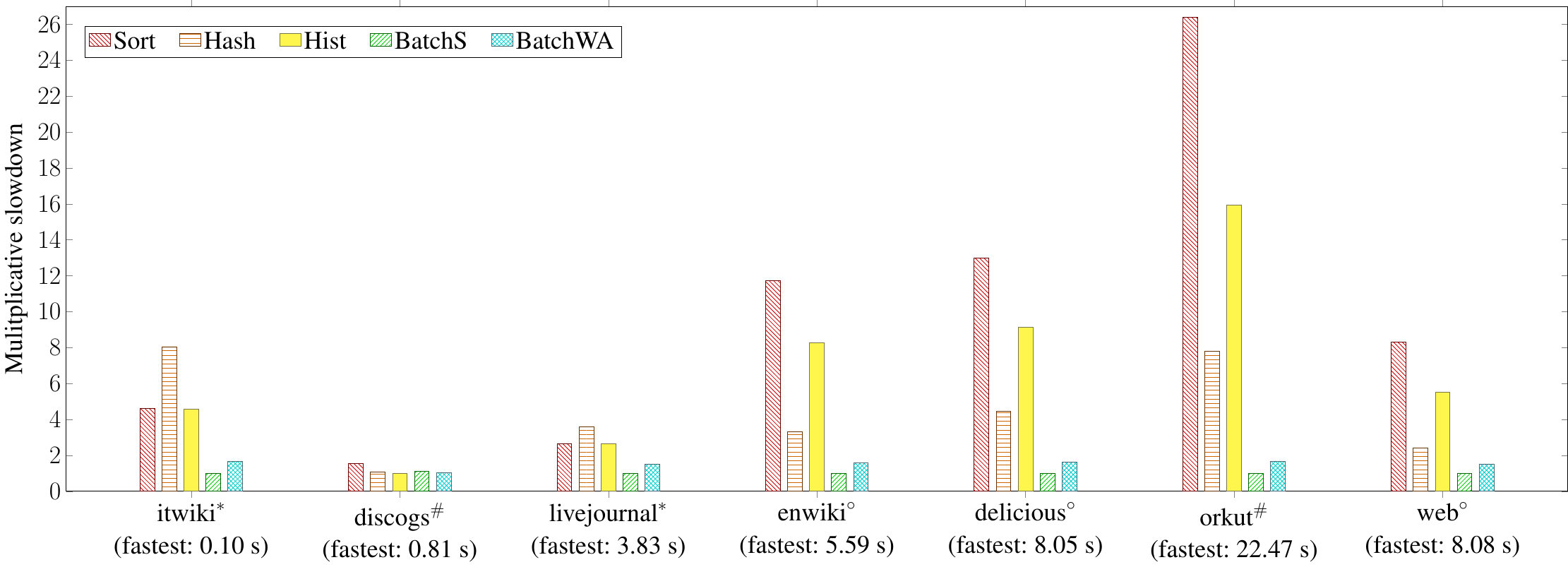}
  \caption{These are the parallel runtimes for butterfly counting in total (using the cache optimization), considering different wedge aggregation methods (butterfly aggregation does not apply). We consider the ranking that produces the fastest runtime for each graph; ${}^*$ refers to side ranking, ${}^\#$ refers to approximate complement degeneracy ranking, and ${}^\circ$ refers to approximate degree ranking. All times are scaled by the fastest parallel time, as parentheses.}\label{inv-fig-count-tot}
\end{figure*}

Also, Figure~\ref{inv-table-serial} shows the runtimes for the
sequential counting implementations using the cache optimization, as
well as runtimes for implementations from previous works, all of which
were tested on the same machine.  Our configurations achieve parallel
speedups of between 5.7--18.0x over the best sequential implementations
for large enough graphs.\footnote{By ``large enough,'' we mean graphs
  for which the sequential counting algorithms take more than 2
  seconds to complete.} We improve upon PGD~\cite{AhmedNRDW17} by 289.5--5170x.

\begin{figure*}[!t]
 \begin{minipage}{.44\textwidth}
 \vspace{\baselineskip}
\centering
\includegraphics[width=\textwidth, page=8]{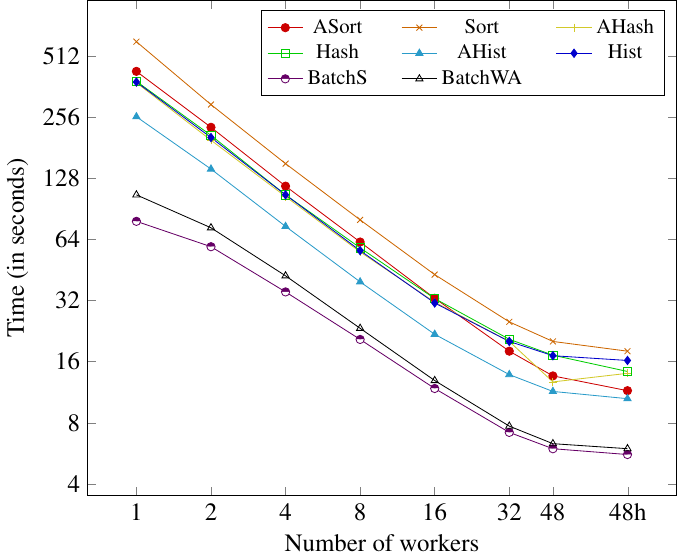}
  \caption{These are the runtimes for butterfly counting per vertex on livejournal using side ranking and using the cache optimization, over different numbers of threads. The self-relative speedups are between 14.1--37.4x.}\label{inv-fig-self-vert}
   \end{minipage}\hspace{1cm}
 \begin{minipage}{.44\textwidth}
   \vspace{\baselineskip}
  \centering
\includegraphics[width=\textwidth, page=7]{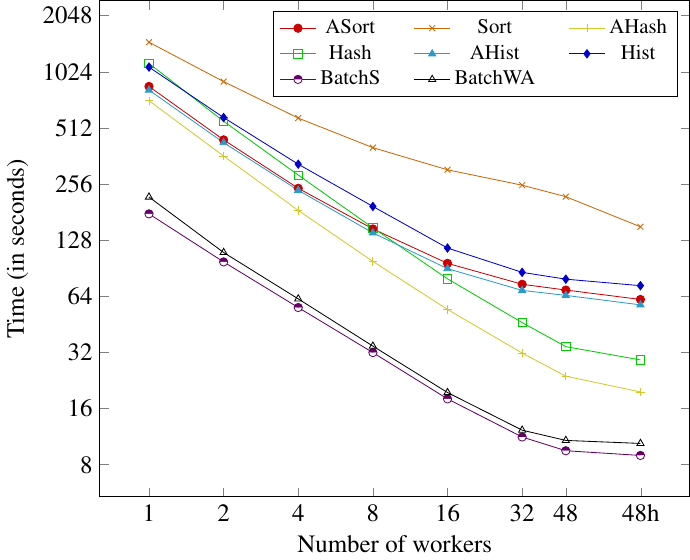}
  \caption{These are the runtimes for butterfly counting per edge on livejournal using approximate degree ranking and using the cache optimization, over different numbers of threads. The self-relative speedups are between 9.8--38.9x.}\label{inv-fig-self-edge}
   \end{minipage}
\end{figure*}

Figures \ref{inv-fig-self-vert} and \ref{inv-fig-self-edge} show our
self-relative speedups on livejournal for per-vertex and per-edge counting,
respectively, using the cache optimization. Across all rankings, on
livejournal, we achieve self-relative speedups between 14.1--37.4x for
per-vertex counting and between 9.8--38.9x for per-edge counting.

\subsubsection{Ranking}
Figure \ref{inv-fig-rank} shows the runtimes for butterfly counting
per vertex for different rankings using the simple batching method
with the cache optimization. The times are normalized to the time for
the fastest ranking for each dataset.  Side ordering outperforms the
other rankings for itwiki and livejournal, while approximate
complement degeneracy, approximate degree, and degree orderings
outperform side ordering for discogs, enwiki, delicious, orkut, and
web.

\begin{figure*}[!t]
  \centering
\includegraphics[width=\textwidth, page=6]{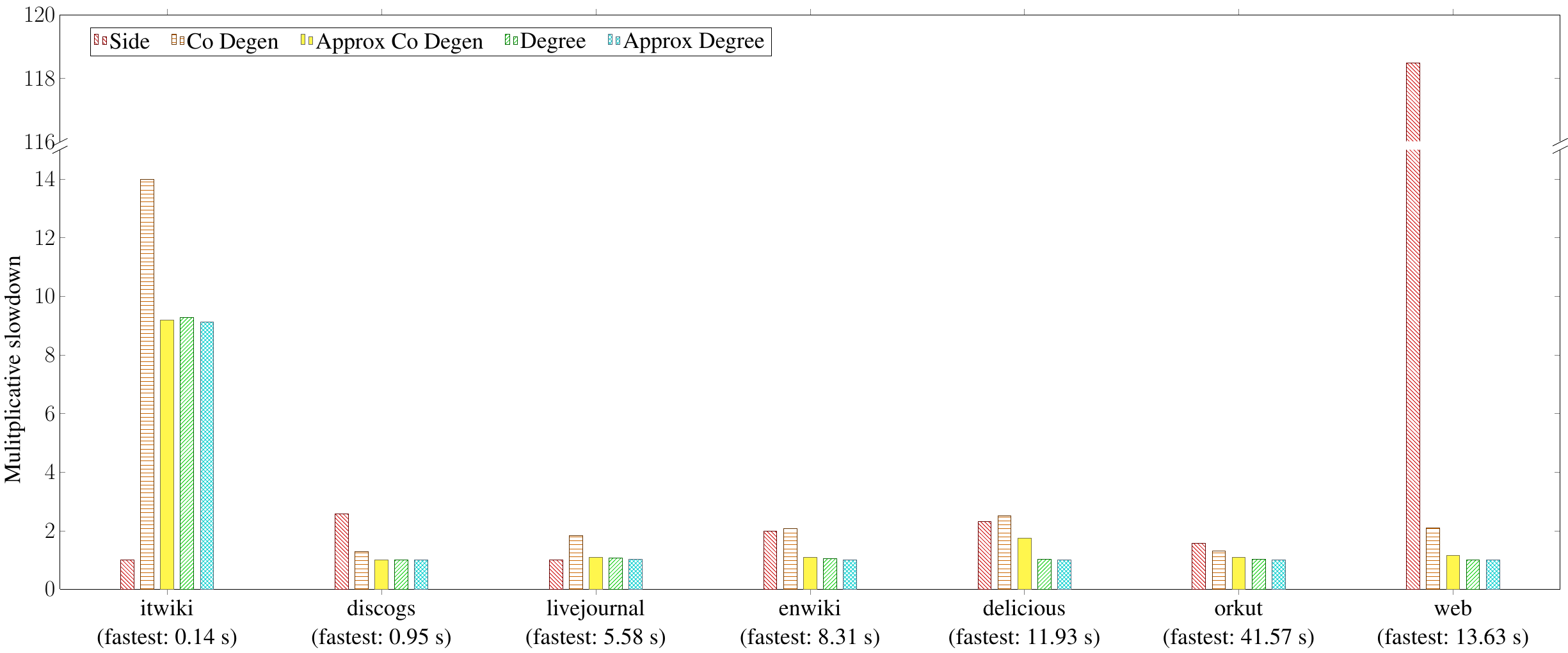}
  \caption{These are the parallel runtimes for butterfly counting per vertex (using the cache optimization), considering different rankings. We use simple batching as our wedge aggregation method. All times are scaled by the fastest parallel time, as indicated in parentheses. Moreover, the time taken to rank each graph is included in the runtimes.}\label{inv-fig-rank}
\end{figure*}

\subsubsection{Approximate counting}
Figure \ref{fig-approx} shows runtimes with the cache optimization for
both colorful sparsification and edge sparsification on orkut, as well
as the corresponding single-threaded times. We see that over a variety
of probabilities $p$ we achieve self-relative speedups between
5.4--18.9x.

\begin{figure*}[!t]
  \centering
\includegraphics[width=.4\textwidth, page=11]{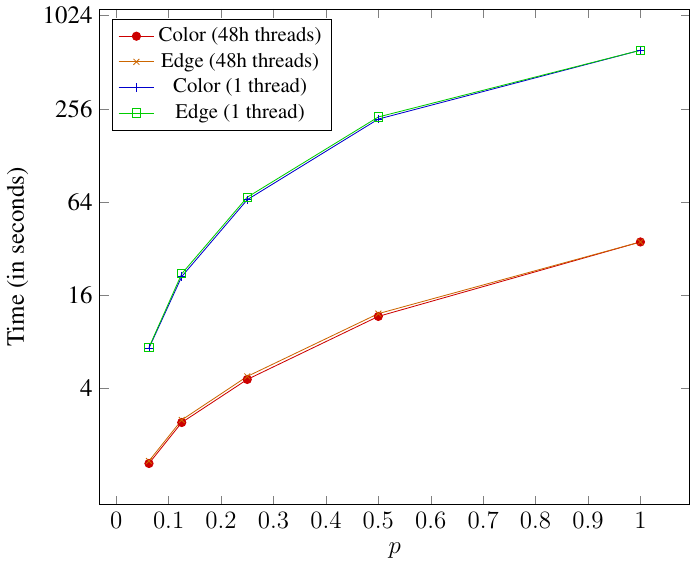}
  \caption{These are the runtimes for colorful sparsification and edge sparsification over different probabilities $p$ (using the cache optimization). We considered both the runtimes on 48 cores hyperthreaded and on a single thread. We ran these algorithms on orkut, using simple batch aggregation and side ranking. }\label{inv-fig-approx}
\end{figure*}

%% file: related.tex
\section{Related Work}

There have been several sequential algorithms designed for butterfly
counting and peeling. Wang \textit{et al.}~\cite{WaFuCh14} propose the
first algorithm for butterfly counting over vertices in $\BigO{\sum_{v
    \in V} \text{deg}(v)^2}$ work, and Sanei-Mehri \textit{et
  al.}~\cite{SaSaTi18} introduce a practical speedup by choosing the
vertex partition with fewer wedges to iterate over. Sanei-Mehri
\textit{et al.}~\cite{SaSaTi18} also introduce approximate counting
algorithms based on sampling and graph sparsification. Later, Zhu
\emph{et al.}~\cite{Zhu2018} present a sequential algorithm for
counting over vertices based on ordering vertices (although they do
not specify which order) in $\BigO{\sum_{v \in V} \text{deg}(v)^2}$
work. They extend their algorithm to the external-memory setting and
also design sampling algorithms.
Chiba and Nishizeki's~\cite{ChNi85} original work on
counting $4$-cycles in general graphs 
applies directly to butterfly counting in bipartite graphs and has a
better work complexity. Chiba and Nishizeki~\cite{ChNi85} use a
ranking algorithm that counts the total number of $4$-cycles in a
graph in $\BigO{\alpha m}$ work, where $\alpha$ is the arboricity of
the graph. While they only give a total count in their work, their
algorithm can easily be extended to obtain counts per-vertex and
per-edge in the same time complexity.  Butterfly counting using degree
ordering was also described by Xia~\cite{Xia2016}.  Sariy\"{u}ce and
Pinar~\cite{SaPi18} introduce algorithms for butterfly counting over
edges, which similarly takes $\BigO{\sum_{v \in V} \text{deg}(v)^2}$
work. Zou \cite{Zou16} develop the first algorithm for butterfly
peeling per edge, with $\BigO{m^2 + m \cdot \text{max-b}_e}$ work. 
Sariy\"{u}ce and Pinar~\cite{SaPi18} give algorithms for butterfly
peeling over vertices and over edges, which take $\BigO{\text{max-b}_v
  +\sum_{u \in U} \text{deg}(u)^2}$ work and $\BigO{ \text{max-b}_e
  + \sum_{u \in U} \sum_{v_1, v_2 \in N(u)} \allowbreak
  \text{max}(\text{deg}(v_1), \text{deg}(v_2))}$ work, respectively.
Very recently, Wang \textit{et al.}~\cite{wang2020efficient} present a sequential algorithm for butterfly edge peeling that improves over the algorithm by Sariy\"{u}ce and Pinar~\cite{SaPi18} in practice, and uses an index that takes $\BigO{\alpha m}$ space.

In terms of prior work on parallelizing these algorithms, Wang
\textit{et al.}~\cite{WaFuCh14} implement a distributed algorithm
using MPI that partitions the vertices across processors, and each
processor sequentially counts the number of butterflies for vertices
in its partition. They also implement a MapReduce algorithm, but show
that it is less efficient than their MPI-based algorithm. The largest
graph they report parallel times for is the \emph{deli} graph with 140
million edges and $1.8 \times 10^{10}$ butterflies; this is the
delicious tag-item graph in KONECT~\cite{Kunegis13}. On this graph,
they take 110 seconds on 16 nodes, whereas on the same graph we take
5.17 seconds on 16 cores.

Very recently, and independently of our work, Wang \emph{et
  al.}~\cite{Wang2019} describe an algorithm for butterfly counting
using degree ordering, as done in Chiba and Nishizeki~\cite{ChNi85},
and also propose a cache optimization for wedge retrieval.  Their
parallel algorithm is our parallel algorithm with simple batching for
wedge aggregation, except they manually schedule the threads, while we
use the Cilk scheduler. They use their algorithm to speed up
approximate butterfly counting, and also propose an external-memory
variant.

There has been recent work on algorithms for finding subgraphs of size
4 or 5~\cite{Hocevar2014,Elenberg2016,Pinar2017,AhmedNRDW17,Dave2017},
which can be used for butterfly counting as a special case.  Marcus
and Shavitt~\cite{Marcus2010} design a sequential algorithm for
finding subgraphs of up to size 4. Hocevar and
Demsar~\cite{Hocevar2014} present a sequential algorithm for counting
subgraphs of up to size 5.  Pinar \emph{et al.}~\cite{Pinar2017} also
present an algorithm for counting subgraphs of up to size 5 based on
degree ordering as done in Chiba and Nishizeki~\cite{ChNi85}.
Elenberg \emph{et al.}~\cite{Elenberg2016}
present a distributed algorithm for counting subgraphs of size
4. Ahmed \emph{et al.}~\cite{AhmedNRDW17} present the PGD
shared-memory framework for counting subgraphs of up to size 4.
The work of their algorithm for counting $4$-cycles is
$\BigO{\sum_{(u,v)\in E}(\text{deg}(v)+\sum_{u' \in
    N(v)}\text{deg}(u'))}$, which is higher than that of our
algorithms.  Aberger et al.~\cite{AbergerLTNOR17} design the
EmptyHeaded framework for parallel subgraph finding based on
worst-case optimal join algorithms~\cite{Ngo2018}.  For butterfly
counting, their approach takes quadratic work. We were unable to
obtain runtimes for EmptyHeaded because it ran out of memory in our
environment.  Dave \emph{et al.}~\cite{Dave2017} present a parallel
method for counting subgraphs of up to size 5 local to each edge. For
counting $4$-cycles, their algorithm is the same as PGD, which we
compare with.  There have also been various methods for approximating
subgraph counts via
sampling~\cite{Jha2015,AhmedWR16,AhmedWR16b,Wang2018,Jain2017,BressanCKLP18,Rossi0A19,Mawhirter2018}.
Finally, there has also been significant work for the past decade on
parallel triangle counting algorithms
(e.g.,~\cite{ShunT2015,Azad2015,Suri2011,Arifuzzaman2013,Park2013,Park14,Cohen2009,Pagh2012,Tangwongsan2013,Tsourakakis2009,Makkar17,Kolda14,Green2015,Green2014,Hu2018,DhBlSh18,Han2018,Zhang2019}
and papers from the annual GraphChallenge~\cite{GraphChallenge}, among
many others).

%% file: conclusion.tex
\section{Conclusion}
We have designed a framework \framework that provides efficient
parallel algorithms for butterfly counting (global, per-vertex, and
per-edge) and peeling (by vertex and by edge).
We have also shown strong theoretical bounds in terms of work and span
for these algorithms. The \framework framework is
built with modular components that can be combined for
practical efficiency. \framework outperforms the best existing
parallel butterfly counting implementations, and we outperform the
fastest sequential baseline by up to 13.6x for butterfly counting and
by up to several orders of magnitude for butterfly peeling.